\documentclass[pdflatex,sn-mathphys-num]{sn-jnl}

\usepackage{graphicx}%
\usepackage{multirow}%
\usepackage{amsmath,amssymb,amsfonts}%
\usepackage{amsthm}%
\usepackage{mathrsfs}%
\usepackage[title]{appendix}%
\usepackage{xcolor}%
\usepackage{textcomp}%
\usepackage{manyfoot}%
\usepackage{booktabs}%
\usepackage{algorithm}%
\usepackage{algorithmicx}%
\usepackage{algpseudocode}%
\usepackage{listings}%

\usepackage{mathtools}
\usepackage{changepage}

\theoremstyle{thmstyleone}%
\newtheorem{theorem}{Theorem}
\newtheorem{proposition}[theorem]{Proposition}%

\newtheorem{lemma}{Lemma}%

\theoremstyle{thmstyletwo}%

\theoremstyle{thmstylethree}%

\begin{document}

\let\vec\boldsymbol
\def\E{\mathbb{E}}
\def\inc{I}
\def\mstar{M_{\bullet}}
\def\real{\operatorname{Re}}
\def\imag{\operatorname{Im}}
\def\L{\mathrm{L}}

\def\vece{\vec{e}}
\def\vech{\vec{h}}
\def\vecrho{\vec{\rho}}
\def\vecz{\vec{z}}
\def\veczeta{\vec{\zeta}}

\def\vecptilde{\vec{\tilde{p}}}
\def\vecrhotilde{\vec{\tilde{\rho}}}

\newcommand{\vechat}[1]{\vec{\hat{#1}}}
\newcommand{\rot}[1]{R_{\vec{\hat{#1}}}}

\newcommand{\A}[1]{\Vert\vec\nu\Vert^{-#1}}
\newcommand{\lapl}[2]{b_{#1/2}^{(#2)}}

\newcommand{\apj}{Astrophysical Journal}

\title[Article Title]{Expansion of the planetary Hamiltonian for small eccentricities and inclinations: a vector formalism}


\author*[]{\fnm{Federico} \sur{Mogavero}}\email{mogavero@math.unipd.it}


\affil[]{\orgdiv{Dipartimento di Matematica “Tullio Levi-Civita”}, \orgname{Università degli Studi di Padova}, \orgaddress{\street{Via Trieste 63}, \postcode{35121} \city{Padova}, \country{Italy}}}



\abstract{We revisit the classical expansion of the planetary Hamiltonian for small eccentricities and mutual inclinations of the orbits, presenting a derivation based entirely on vector formalism. We demonstrate that the secular part of the disturbing function, as well as any other Fourier harmonic, can be expressed in terms of scalar products of vectors lying in the system's invariant plane. Furthermore, we show how to express any such term using the angular momentum and eccentricity vectors of the orbits.}


\maketitle

\section{Introduction}
\label{sect:intro}

A large part of the classical and modern developments in planetary dynamics relies on the expansion of the disturbing function as a power series in certain orbital elements. The disturbing function contains the gravitational interactions between planets and governs the long-term evolution of the orbits. 
Two standard approaches are commonly employed to carry out its expansion. In hierarchical systems, the disturbing function is developed in powers of the semi-major axis ratios, while keeping exact dependence on eccentricities and inclinations \cite{Hansen1855,Hill1875,Tisserand1894}. In contrast, for compact configurations, it is expanded as a series in eccentricities and inclinations, with coefficients that are exact in semi-major axis ratios \cite{Laplace1773,Lagrange1781,Lagrange1782,LeVerrier1855}. 

The development in semi-major axis ratios can be readily carried out in a vector-based formalism, by writing the gravitational potential in terms of Legendre polynomials \cite[e.g.,][]{Tremaine2023}. The situation is different when it comes to the expansion in eccentricities and inclinations, for which one usually employs scalar quantities. A notable exception is the work of \citet{Abdullah2001a,Abdullah2001}, later taken up by \citet{Boue2014}. Abdullah starts from a standard expansion of the disturbing function involving scalar orbital elements and then restates its secular part in terms of the spin-like variables $\vec{h}+\vec{e}$ and $\vec{h}-\vec{e}$, where $\vec{h}$ and $\vec{e}$ are the dimensionless angular momentum and eccentricity vectors of a Keplerian orbit\footnote{See \cite{Rosengren2014} for a detailed historical account on the use of these vectors for perturbed Keplerian motions.}, respectively. The resulting secular Hamiltonian is remarkably compact. Moreover, as it only depends on scalar products involving the vectors $\vec{h}$ and $\vec{e}$, its expression is independent of any reference frame, which is one of the major advantages of a vector-based formalism. 

In this work, we revisit the expansion of the planetary Hamiltonian for small eccentricities and mutual inclinations of the orbits by carrying out a derivation that employs a vector formalism from the beginning. We first show how the expansion is naturally expressed in terms of vectors lying in the invariant plane of the planetary system (Secs.~\ref{sect:direct_part} and \ref{sec:indirect_part}). Then, we explain how to write both its secular part and all other Fourier harmonics -- traditionally referred to as ``inequalities'' in celestial mechanics -- by means of scalar products of such vectors (Secs.~\ref{sect:sec_ham} and \ref{sect:sec_inequalities}). We also express the inequalities in terms of the vectors $\vech$ and $\vece$ (Sec.~\ref{sec:eh}). We conclude by discussing the usefulness of a vector-based formalism. 

\section{Planetary Hamiltonian}
\label{sect:hamiltonian}

The Hamiltonian of a planetary system consisting of a star of mass $m_0$ and $n$ planets of masses $m_1, m_2, \dots, m_n \ll m_0$ can be written using canonical astrocentric variables as follows \citep[and references therein]{Laskar1991}: 
\begin{equation}
\label{eq:Hamiltonian}
H = \underbrace{\sum_{i=1}^n \left( \frac{\Vert \vec{p}_i \Vert^2}{2 \beta_i} 
- \frac{\mu_i \beta_i}{\Vert \vec{r}_i \Vert} \right)}_{H_0}
+ \underbrace{\sum_{1 \leq i < j \leq n} \left( \frac{\vec{p}_i \cdot \vec{p}_j}{m_0} - \frac{G m_i m_j}{\Vert \vec{r}_i - \vec{r}_j \Vert} \right)}_{H_1},
\end{equation}
where $\vec{r}_i$ are the astrocentric position vectors of the planets and $\vec{p}_i$ their barycentric momenta (so $\vec{p}_i \neq m_i \dot{\vec{r}}_i$); $G$ is the gravitational constant, $\beta_i= m_0 m_i/(m_0 + m_i)$ are the reduced planetary masses, and $\mu_i = G(m_0 + m_i)$.

When planets are away from any close encounter, the disturbing function $H_1$ is much smaller than Keplerian part $H_0$. In this case, the Hamiltonian~\eqref{eq:Hamiltonian} is in quasi-integrable form, and one can introduce canonical variables that trivially integrate the Kepler problem. Poincaré considered the following set of modified Delaunay variables in complex form \cite{Laskar1991}: 
\begin{equation}
\label{eq:mod_delaunay}
\begin{aligned}
&\Lambda = \beta\sqrt{\mu a}, \quad &x = \sqrt{\Lambda} \sqrt{1 - \sqrt{1- e^2}} \, \E^{\iota \varpi}, \\
&\lambda = M + \varpi, \quad &y = \sqrt{2 \Lambda} \left(1- e^2\right)^{\frac{1}{4}} \sin(\inc/2) \E^{\iota \Omega}, 
\end{aligned}
\end{equation}
where $a$ is the semi-major axis of a Kepler orbit, $e$ the eccentricity, $I$ the inclination, $\varpi$ the longitude of perihelion, $\Omega$ the longitude of the ascending node, $\lambda$ the mean longitude, and $M$ the mean anomaly. Throughout the article, $\iota$ represents the imaginary unit, $\E$ stands for the exponential operator, and the overbar denotes the complex conjugate. In Eq.~\eqref{eq:mod_delaunay}, $(\Lambda,\lambda), (x, -\iota \overline{x})$ and $(y, -\iota \overline{y})$ are momentum-coordinate pairs of canonical variables. With this choice of variables, the Keplerian part of the Hamiltonian reads as $H_0~=~-\sum_{i=1}^n \mu_i^2 \beta_i^3 / 2 \Lambda_i^2$, with mean motions $n_i = \dot{\lambda}_i = \sqrt{\mu_i/a_i^3}$, and the variables $\Lambda,x,y$ are integrals of motion for the Kepler problem. 

The complex variables $x, y$ in Eqs.~\eqref{eq:mod_delaunay} are related to the classical non-canonical elliptic elements $z=e\E^{\iota \varpi}$ and $\zeta=\sin(I/2)\E^{\iota \Omega}$ by the following identities:
\begin{equation}
\label{eq:zetaz_xy}
\begin{aligned}
&x = z \sqrt{\Lambda} \left( 1 + \sqrt{1- z \bar{z}}\right)^{-1/2}, 
&z = x \sqrt{\frac{2}{\Lambda}} \left( 1 - \frac{x \bar{x}}{2 \Lambda} \right)^{1/2},  \\
&y = \zeta \sqrt{2 \Lambda} \left( 1 - z \bar{z} \right)^{1/4}, 
&\zeta = y \frac{1}{\sqrt{2 \Lambda}} \left( 1 - \frac{x \bar{x}}{\Lambda} \right)^{-1/2} . 
\end{aligned}
\end{equation}
Hereafter, we will denote $\phi = (1- z \bar{z})^{1/2}$, $\chi = (1- \zeta \bar{\zeta})^{1/2}$, and $\psi = (1+\phi)^{-1}$.

As it will be shown in Sect.~\ref{sect:position_vector}, the position vector $\vec{r}$ can be expressed in a natural way as a function of the eccentric longitude $F = E + \varpi$, where $E$ stands for the eccentric anomaly. When written in terms of $F$, Kepler's equation reads as: 
\begin{equation}
\label{eq:kepler}
\lambda = F - \imag \left( \bar{z} \E^{\iota F} \right) ,
\end{equation}
while the heliocentric distance $r = \Vert \vec{r} \Vert$ is given by:
\begin{equation}
\label{eq:distance}
r = a \left[ 1 - \real \left( \bar{z} \E^{\iota F} \right) \right] .
\end{equation} 

\section{Direct part of the disturbing function}
\label{sect:direct_part}

We aim at expanding the direct part of the disturbing function $H_1$ related to the interaction of a pair of planets with heliocentric positions $\vec{r}$ and $\vec{r}'$, that is, 
\begin{equation}
\label{eq:direct_part}
\frac{a'}{\Delta} \vcentcolon= \frac{a'}{\Vert \vec{r} - \vec{r}' \Vert} = \frac{1}{\Vert \vecrho' - \alpha \vecrho \Vert},
\end{equation}
where $\vec{\rho} = \vec{r}/a$ is the dimensionless position vector and $\alpha = a/a'$ is the semi-major axis ratio. Hereafter, primed variables refer to the outer orbit, so $\alpha \leq 1$. 

\subsection{Position vector as a function of the orbital elements}
\label{sect:position_vector}

We first search for an expression of the position vector $\vec\rho$ in terms of orbital elements. Throughout our derivations, we will consider an inertial reference frame $\mathcal{I}$ with a right-handed basis of unit vectors $\vec{\hat{\imath}}, \vec{\hat{\jmath}}, \vec{\hat{k}}$\footnote{Throughout the paper, unit vectors are denoted by an overhat. The symbols $\times$ and $\cdot$ indicate cross and scalar products of vectors, respectively.}. The vertical axis $\vec{\hat{k}}$ is assumed to be aligned with the total angular momentum of the planetary system; it is therefore perpendicular to the system's invariant planet. With respect to such reference frame, the generic vector $\vec{v} \in \mathbb{R}^3$ will be written as:
\begin{equation}
\vec{v} = 
\begin{bmatrix} v_\imath \\ v_\jmath \\ v_k \end{bmatrix}_\mathcal{I} \vcentcolon= 
v_\imath \vec{\hat{\imath}} + v_\jmath \vec{\hat{\jmath}} + v_k \vec{\hat{k}} .
\end{equation}
We will also consider a (non-inertial) reference frame $\mathcal{O}$ attached to the instantaneous orbit of a given planet, with axes $\vec{\hat{e}}, \vec{\hat{h}} \times \vec{\hat{e}}, \vec{\hat{h}}$: the unit vector $\vec{\hat{e}}$ is directed towards the pericenter and the vertical axis $\vec{\hat{h}}$ is aligned with the angular momentum of the orbit. 

We start by considering the conservation of the Laplace–Runge–Lenz vector $\vec{A}= \vec{p} \times \vec{L} - \mu \beta^2 \vec{\hat{r}}$, where $\vec{L} = \vec{r} \times \vec{p}$ denotes the angular momentum of the orbit and $\vec{\hat{r}} = \vec{r}/r$ \citep{Goldstein2002}. Defining the eccentricity vector as $\vece = \vec{A} / \mu \beta^2 = e \vec{\hat{e}}$, one has:
\begin{equation}
\label{eq:runge_lenz_2}
\vecrho = \rho \vec{\hat{r}} = \vecptilde \times \vech - \rho \vece,
\end{equation}
where we denoted $\vecptilde = \rho \vec{p} / \beta n a$; $\vech = \vec{L} / \Lambda = \phi \vec{\hat{h}}$ is the dimensionless angular momentum and $\rho = \Vert \vecrho \Vert$. The vector $\vecptilde$ can be readily written in the reference frame attached to the orbit:
\begin{equation}
\vecptilde = 
\begin{bmatrix} -\sin{E} \\ \phi \cos{E} \\ 0 \end{bmatrix}_\mathcal{O} = 
\begin{bmatrix} -\sin{E} \\ \cos{E} \\ 0 \end{bmatrix}_\mathcal{O} -
\psi e \cos(E) \begin{bmatrix} 0 \\ e \\ 0 \end{bmatrix}_\mathcal{O}.
\end{equation}
We restate this expression as follows: 
\begin{equation}
\label{eq:dimensionless_momentum}
\vecptilde = \vec{\hat{h}} \times \left[ \vecrhotilde - \psi (\vecrhotilde \cdot \vece) \vece \right],
\end{equation}
where we defined the vector $\vecrhotilde$ as:
\begin{equation}
\vecrhotilde = 
\begin{bmatrix} \cos{E} \\ \sin{E} \\ 0 \end{bmatrix}_\mathcal{O} .
\end{equation}
From Eqs.~\eqref{eq:distance}, \eqref{eq:runge_lenz_2}, and \eqref{eq:dimensionless_momentum}, we obtain the following expression for the dimensionless position vector:
\begin{equation}
\label{eq:rho}
\vecrho = \phi \vecrhotilde - (1 - \psi \vecrhotilde \cdot \vece) \vece .
\end{equation}

Although the expression of the vector $\vecrhotilde$ is simple in the reference frame of the orbit, throughout the paper we will make use of the inertial reference frame $\mathcal{I}$. Therefore, we consider the following relation: 
\begin{equation}
\label{eq:rho_tilde_inertial_frame}
\vecrhotilde = \rot{\vec{h}}(\omega + E) \rot{\vec{\zeta}}(\inc) \rot{\vec{k}}(\Omega) \, \vec{\hat{\imath}} ,
\end{equation}
where $\rot{\vec{u}}(\theta) \vec{v}$ denotes the rotation of a vector $\vec{v} \in \mathbb{R}^3$ around the unit vector $\vec{\hat{u}}$ by an angle $\theta$; $\omega$ is the argument of the pericenter and the vector $\veczeta = \sin(\inc/2) \vechat{\zeta}$ points towards the ascending node of the orbit, that is,
\begin{equation}
\label{eq:zeta}
\veczeta = \sin(\inc/2)
\begin{bmatrix} \cos{\Omega} \\ \sin{\Omega} \\ 0 \end{bmatrix}_\mathcal{I} .
\end{equation}
We recall Rodrigues' rotation formula \cite{Euler1776,Rodrigues1840,Gibbs1901}:
\begin{equation}
\label{eq:rotation}
\begin{aligned}
\rot{\vec{u}}(\theta) \vec{v} &= \vec{v} + (1-\cos \theta) \, \vechat{u} \times (\vechat{u} \times \vec{v}) + \sin(\theta) \, \vechat{u} \times \vec{v} \\
&= \cos(\theta) \, \vec{v} + (1-\cos \theta) (\vechat{u} \cdot \vec{v}) \, \vechat{u}, + \sin(\theta) \, \vechat{u} \times \vec{v} ,
\end{aligned}
\end{equation}
The rotation of a generic vector $\vec{v}$ around the line of nodes by an angle $\inc$ reads as follows:
\begin{equation}
\label{eq:rotation_zeta}
\rot{\vec{\zeta}}(\inc) \, \vec{v} = 
\vec{v} + 2 \, \veczeta \times (\veczeta \times \vec{v}) + 2 \chi \, \veczeta \times \vec{v} . 
\end{equation}
In particular, the direction of the angular momentum vector is given by:
\begin{equation}
\label{eq:angular_momentum}
\vechat{h} = \rot{\vec{\zeta}}(\inc) \, \vechat{k} = 
(1-2 \veczeta \cdot \veczeta) \, \vechat{k} + 2 \chi \, \veczeta \times \vechat{k}. 
\end{equation}
To explicitly perform the rotations in Eq.~\eqref{eq:rho_tilde_inertial_frame}, we present the following lemma.
\begin{lemma}[]
\label{lemma1}
For any $\vec{v} \in \mathbb{R}^3$ and $\theta \in [0,2\pi)$ one has: 
\begin{equation}
\label{eq:lemma}
\rot{\vec{h}}(\theta) \rot{\vec{\zeta}}(\inc) \vec{v} = \rot{\vec{\zeta}}(\inc) \rot{\vec{k}}(\theta) \vec{v} .
\end{equation}
\end{lemma}
The demonstration is a straightforward but tedious application of Eqs.~\eqref{eq:rotation}, \eqref{eq:rotation_zeta}, and \eqref{eq:angular_momentum}. In Appendix~\ref{sec:A0} we present a shorter derivation exploiting the connection between spatial rotations and quaternions. Applying the lemma to Eq.~\eqref{eq:rho_tilde_inertial_frame}, one obtains:
\begin{equation}
\label{eq:rho_tilde_2}
\begin{aligned}
\vecrhotilde &= 
\rot{\vec{\zeta}}(\inc) \rot{\vec{k}}(F) \, \vec{\hat{\imath}} = \rot{\vec{\zeta}}(\inc) \vecrho_0 \\
&= \vecrho_0 + 2 \, \veczeta \times (\veczeta \times \vecrho_0) + 2 \chi \, \veczeta \times \vecrho_0 ,
\end{aligned}
\end{equation}
where $\vecrho_0$ is the dimensionless position vector in the circular and coplanar limit:
\begin{equation}
\label{eq:rho0}
\vecrho_0 = \rot{\vec{k}}(F) \, \vec{\hat{\imath}} = 
\begin{bmatrix} \cos{F} \\ \sin{F} \\ 0 \end{bmatrix}_\mathcal{I} . 
\end{equation}
As Eq.~\eqref{eq:rho_tilde_2} shows, with the introduction of the eccentric longitude $F$, the vector $\vecrhotilde$ does not depend on the eccentricity vector $\vece$ anymore: it only depends on the orientation of the orbital plane through the vector $\veczeta$. So the eccentricity and inclination dependence of the position vector $\vecrho$ neatly splits up in Eq.~\eqref{eq:rho}. 

In a similar manner, we obtain the following expression for the eccentricity vector:
\begin{equation}
\begin{aligned}
\label{eq:ecc_vec}
\vece &= e \, 
\rot{\vec{h}}(\omega) \rot{\vec{\zeta}}(\inc) \rot{\vec{k}}(\Omega) \, \vec{\hat{\imath}}
= e \, \rot{\vec{\zeta}}(\inc) \rot{\vec{k}}(\varpi) \, \vec{\hat{\imath}} = \rot{\vec{\zeta}}(\inc) \, \vecz \\
&= 
\vecz + 2 \, \veczeta \times (\veczeta \times \vecz) + 2 \chi \, \veczeta \times \vecz , 
\end{aligned}
\end{equation}
where we defined the vector $\vecz$ as:
\begin{equation}
\label{eq:z}
\vecz = e \, \rot{\vec{k}}(\varpi) \, \vec{\hat{\imath}} = e
\begin{bmatrix} \cos{\varpi} \\ \sin{\varpi} \\ 0 \end{bmatrix}_\mathcal{I} .
\end{equation}
Noting that $\vecrhotilde \cdot \vece = \vecrho_0 \cdot \vecz$, from Eqs.~\eqref{eq:rho} and \eqref{eq:rho_tilde_2} we obtain the following final expression for the dimensionless position vector:
\begin{equation}
\label{eq:rho_2}
\begin{aligned}
\vecrho &= \vecrho_0 + \vec\delta, \\
\vec\delta &= (\phi-1) \vecrho_0 + 2 \phi \, [ \veczeta \times (\veczeta \times \vecrho_0) +  \chi \,\veczeta \times \vecrho_0 ] - ( 1 - \psi \vecz \cdot \vecrho_0) \vece
\end{aligned}
\end{equation}
The vector $\vecrho_0$ defines the position of the planet along its instantaneous Keplerian orbit, while the vectors $\vecz$ and $\veczeta$ encode the two slow degrees of freedom related to eccentricity and inclination, respectively. The three vectors relate in a natural and standard way to the complex variables $\E^{\iota F}$, $z$ and $\zeta$, which are often employed in the expansion of the planetary Hamiltonian: 
\begin{equation}
\label{eq:vectors_complex_vars}
\vecrho_0 = \begin{bmatrix} \real(\E^{\iota F}) \\ \imag(\E^{\iota F}) \\ 0 \end{bmatrix}_\mathcal{I} , \quad
\vecz = \begin{bmatrix} \real(z) \\ \imag(z) \\ 0 \end{bmatrix}_\mathcal{I} , \quad
\veczeta = \begin{bmatrix} \real(\zeta) \\ \imag(\zeta) \\ 0 \end{bmatrix}_\mathcal{I} .
\end{equation}
We stress that $\vecrho_0$, $\vecz$, and $\veczeta$ are all parallel to the invariant plane of the planetary system.
When the eccentricities and mutual inclinations of the orbits are small, the norm of the vector $\vec\delta$ is much smaller than that of $\vecrho_0$; this property will be used as a basis for the expansion of the disturbing function presented in Sect.~\ref{sect:expansion_direct_part}.

Finally, it can be useful to consider the decomposition $\vec\delta = \vec\delta_\parallel + \vec\delta_\perp$ where:
\begin{equation}
\label{eq:vecdelta_decomposition}
\begin{aligned}
\vec\delta_\parallel &= (\phi-1) \vecrho_0 + 2 \phi \, \veczeta \times (\veczeta \times \vecrho_0) - ( 1 - \psi \vecz \cdot \vecrho_0) \vece_\parallel , \\
\vec\delta_\perp &= 2 \phi \chi \, \veczeta \times \vecrho_0 - ( 1 - \psi \vecz \cdot \vecrho_0) \vece_\perp .
\end{aligned}
\end{equation}
Here, $\vec\delta_\parallel$ is the component of $\vec\delta$ that is parallel to the invariant plane, while $\vec\delta_\perp$ is the component that is perpendicular to it; we decomposed in a similar way the eccentricity vector as $\vece = \vece_\parallel + \vece_\perp$ where:
\begin{equation}
\label{eq:vece_decomposition}
\begin{aligned}
\vece_\parallel &= \vecz + 2 \, \veczeta \times (\veczeta \times \vecz) , \\
\vece_\perp &= 2 \chi \, \veczeta \times \vecz .
\end{aligned}
\end{equation}

\subsection{Expansion for small eccentricities and inclinations}
\label{sect:expansion_direct_part}

Using Eqs.~\eqref{eq:rho_2}, the direct part of the disturbing function [Eq.~\eqref{eq:direct_part}] can be written as follows:
\begin{equation}
\begin{aligned}
\label{eq:direct_part_2}
\frac{a'}{\Delta} = \frac{1}{\Vert \vec\nu + \vec\epsilon \Vert} ,
\end{aligned}
\end{equation}
where $\vec\nu = \vecrho'_0 - \alpha \vecrho_0$ and $\vec\epsilon = \vec\delta' - \alpha \vec\delta$.
In the regime of small eccentricities and mutual inclinations of the orbits, one has $\Vert \vec\epsilon \Vert \ll \Vert \vec\nu \Vert$. Given the vector form of Eq.~\eqref{eq:direct_part_2}, we can introduce the classical expansion of the gravitational potential in terms of Legendre polynomials\footnote{Interestingly enough, Legendre polynomials are usually introduced when expanding the inverse of the mutual distance in the limit of a small semi-major axis ratio (for arbitrary eccentricities and inclinations).}:
\begin{equation}
\begin{aligned}
\label{eq:legendre}
\frac{a'}{\Delta} 
&= \sum_{l=0}^{+\infty} \frac{\Vert\vec\epsilon\Vert^{l}}{\Vert\vec\nu\Vert^{l+1}} 
P_l \! \left( \frac{\vec\epsilon \cdot \vec\nu}{\Vert\vec\epsilon\Vert \Vert\vec\nu\Vert} \right)  \\
&= \sum_{l=0}^{+\infty} \sum_{k=0}^{\lfloor l/2 \rfloor} 
\frac{(-1)^{l-k}}{2^l} \binom{l}{k} \binom{2l-2k}{l} 
\frac{\Vert\vec\epsilon\Vert^{2k}}{\Vert\vec\nu\Vert^{1+2(l-k)}} (\vec\nu \cdot \vec\epsilon)^{l-2k} ,
\end{aligned}
\end{equation}
where $P_l$ is the Legendre polynomial of degree $l$. The norm of the vector $\vec\nu$ only depends on the difference of the eccentric longitudes, that is,
\begin{equation}
\label{eq:vector_nu_norm}
\Vert\vec\nu\Vert = \sqrt{1 + \alpha^2 -2\alpha\cos(F'-F)} .
\end{equation}
We notice that, using the decompositions~\eqref{eq:vecdelta_decomposition} and \eqref{eq:vece_decomposition}, and defining the vectors $\vec\epsilon_\parallel = \vec\delta'_\parallel - \alpha \vec\delta_\parallel$ and $\vec\epsilon_\perp = \vec\delta'_\perp - \alpha \vec\delta_\perp$, one has $\Vert\vec\epsilon\Vert^{2} = \Vert\vec\epsilon_\parallel\Vert^{2} + \Vert\vec\epsilon_\perp\Vert^{2}$ and $\vec\nu \cdot \vec\epsilon = \vec\nu \cdot \vec\epsilon_\parallel$. 

Using vector algebra identities in a systematic way \cite{Liang2007}, expansion~\eqref{eq:legendre} can be fully written in terms of scalar products that involve vectors belonging to the following set:
\begin{equation}
V = \{ \vecrho_0, \vecrho_0', \vecz, \vecz', \veczeta, \veczeta' \} .
\end{equation}
Although it is straightforward, using the multinomial theorem, to obtain an explicit (but cumbersome) expansion, Eq.~\eqref{eq:legendre} can be easily implemented in a computer algebra system that allows for manipulation of symbolic vectors. Therefore, in the following we focus on the structure of the different terms that appear in the expansion. Among all the possible scalar products, we have three terms of degree 0 in eccentricity and inclinations of the orbits:
\begin{equation}
\label{eq:terms_degree_0}
\vec\rho_0 \cdot \vecrho_0 = 1, \quad \vec\rho_0' \cdot \vecrho_0' = 1, \quad 
\vec\rho_0 \cdot \vecrho_0' = \cos(F'-F) .
\end{equation}
We then have eight different terms of degree 1, which depend on $F,F'$:
\begin{equation}
\label{eq:terms_degree_1}
\vec\rho_0 \cdot \vec{u}_i, \quad \vec\rho_0' \cdot \vec{u}_i, \quad \mathrm{with}\enspace  
\vec{u}_i \in V_1 = \{ \vecz, \vecz', \veczeta, \veczeta' \} .
\end{equation}
Finally, we have ten scalar products of degree 2, which do not depend on $F,F'$:
\begin{equation}
\label{eq:terms_degree_2}
\vec{u}_i \cdot \vec{u}_j, \quad \mathrm{with}\enspace \vec{u}_i, \vec{u}_j \in V_1 .
\end{equation}

We notice that the quantities $\phi = \sqrt{1-\vecz \cdot \vecz}$, $\chi = \sqrt{1-\veczeta \cdot \veczeta}$ and $\psi = (1+\phi)^{-1}$ enter expansion~\eqref{eq:legendre} through Eq.~\eqref{eq:rho_2}. They can be developed as power series in the scalar products $\vecz \cdot \vecz$ or $\veczeta \cdot \veczeta$, resulting in a direct part that is completely expanded in eccentricities and inclinations. It is then possible to collect together all the terms of the same total degree $d$. By doing that, expansion~\eqref{eq:legendre} can be formally restated as follows:
\begin{equation}
\begin{aligned}
\label{eq:legendre_trunc}
\frac{a'}{\Delta} = \sum_{d=0}^{+\infty}
\sum_\mathcal{N}
\Gamma_{\mathcal{N}}
\frac{\alpha^l}{\Vert\vec\nu\Vert^s}
(\vecrho_0 \cdot \vecrho_0')^m 
\left( \prod_{i=1}^4 (\vec\rho_0 \cdot \vec{u}_i)^{p_i} (\vec\rho_0' \cdot \vec{u}_i)^{p'_i} \right)
\! \prod_{1\leq j\leq k\leq 4} \! (\vec{u}_j \cdot \vec{u}_k)^{q_{jk}} .
\end{aligned}
\end{equation}
Here we fixed $(\vec{u}_1,\vec{u}_2,\vec{u}_3,\vec{u}_4) = (\vecz, \vecz', \veczeta, \veczeta')$ and $\Gamma_{\mathcal{N}}$ is a rational number that depends on the tuple of non-negative integers $\mathcal{N} = (s,l,m,p_1,\dots,p_4,p'_1,\dots,p'_4,q_{11},\dots,q_{44})$. The summation over $\mathcal{N}$ depends on $d$ and verify the following constraints:
\begin{equation}
\label{eq:sum_N}
\begin{aligned}
\min(3,2d+1) \leq s \leq 2d+1, \quad
0 \leq l \leq 2d, \quad
0 \leq m \leq \lfloor d/2 \rfloor, \\
\sum_{i=1}^4 \left(p_i + p_i'\right) + \sum_{1\leq j\leq k\leq 4} 2 \, q_{jk} = d .
\end{aligned}
\end{equation}
Not all the coefficients $\Gamma_{\mathcal{N}}$ that are allowed by these constrains are different from zero. For example, each addend of the summation~\eqref{eq:legendre_trunc} is invariant under the transformation $\veczeta, \veczeta' \rightarrow -\veczeta, -\veczeta'$, that is, $p_3 + p_4 + \sum_{j=1}^2 (q_{j3} + q_{j4})$ is an even number. This is consistent with the invariance of the direct part of the disturbing function under reflection of the position vectors $\vecrho, \vecrho'$ with respect to the invariant plane\footnote{Expansion~\eqref{eq:legendre_trunc} is also clearly invariant under any spatial rotation of the position vectors, as it is fully written in terms of scalar products.}.

In Appendix~\ref{sec:A1_direct_part_2} we give, as an example, the expansion of the direct part of the disturbing function to degree 2 in eccentricities and inclinations, which has been obtained by implementing Eq.~\eqref{eq:legendre} in the computer algebra system TRIP \cite{Gastineau2011,TRIP}.

\section{Indirect part of the disturbing function}
\label{sec:indirect_part}

The indirect part of the disturbing function $H_1$ related to the interaction of a pair of planets with barycentric momenta $\vec{p}$ and $\vec{p}'$ is: 
\begin{equation}
\label{eq:indirect_part}
\frac{\vec{p} \cdot \vec{p}'}{m_0} = \frac{\beta \beta' n n' a a'}{m_0} \frac{\vecptilde \cdot \vecptilde'}{\rho \rho'} ,
\end{equation}
with $\rho = 1 - \vecrho_0 \cdot \vecz$ [see Eq.~\eqref{eq:distance}]. 
Similarly to the direct part, we show here how Eq.~\eqref{eq:indirect_part} can be expressed in a way that only involves vectors belonging to the set $V$. 

Equation~\eqref{eq:dimensionless_momentum} can be restated as $\vecptilde = \vec{\hat{h}} \times \vec{q}$ with $\vec{q} = \vecrhotilde - \psi (\vecrhotilde \cdot \vece) \vece$. Therefore, 
\begin{equation}
\vecptilde \cdot \vecptilde' = (\vechat{h} \cdot \vechat{h}')(\vec{q} \cdot \vec{q}') - (\vechat{h} \cdot \vec{q}')(\vechat{h}' \cdot \vec{q}) .
\end{equation}
Using Eqs.~\eqref{eq:rho_tilde_2} and \eqref{eq:vece_decomposition}, we consider the decomposition $\vec{q} = \vec{q}_\parallel + \vec{q}_\perp$ where:
\begin{equation}
\label{eq:vecq}
\begin{aligned}
\vec{q}_\parallel &= \vecrho_0 + 2 \, \veczeta \times (\veczeta \times \vecrho_0) - \psi (\vecrho_0 \cdot \vecz) \vece_\parallel, \\
\vec{q}_\perp &= 2 \chi \, \veczeta \times \vecrho_0 - \psi (\vecrho_0 \cdot \vecz) \vece_\perp .
\end{aligned}
\end{equation}
Using Eq.~\eqref{eq:angular_momentum}, one gets: 
\begin{equation}
\begin{aligned}
\vechat{h} \cdot \vechat{h}' &= (1-2 \veczeta \cdot \veczeta)(1-2 \veczeta' \cdot \veczeta') + 4 \chi \chi' \veczeta \cdot \veczeta', \\
\vec{q} \cdot \vec{q}' &= \vec{q}_\parallel \cdot \vec{q}'_\parallel + \vec{q}_\perp \cdot \vec{q}'_\perp, \\
(\vechat{h} \cdot \vec{q}')(\vechat{h}' \cdot \vec{q}) &= [(1-2 \veczeta \cdot \veczeta)\vec{q}'_\perp + 2 \chi \, \vec{q}'_\parallel \times \veczeta] \cdot [(1-2 \veczeta' \cdot \veczeta')\vec{q}_\perp + 2 \chi' \, \vec{q}_\parallel \times \veczeta'].
\end{aligned}
\end{equation}
Therefore, using vector algebra identities, Eq.~\eqref{eq:indirect_part} can be restated in terms of scalar products of vectors belonging to the set $V$. The indirect part can be eventually fully expanded in eccentricities and inclinations by developing the quantities $\psi = (1+\sqrt{1-\vecz \cdot \vecz})^{-1}$ and $\chi = \sqrt{1-\veczeta \cdot \veczeta}$, as well as $\psi'$ and $\chi'$. 

\section{Secular disturbing function}
\label{sect:sec_ham}

Let us denote as $h_1$ the contribution to the disturbing function $H_1$ raised by the interaction of a given couple of planets:
\begin{equation}
\label{eq:h1}
h_1 = \frac{\vec{p} \cdot \vec{p}'}{m_0} - \frac{G m m'}{\Delta} .
\end{equation}
In this section, we show how to compute the secular disturbing function, which consists in the null-frequency component of $h_1$ when expanded as a Fourier series with respect to the mean longitudes, that is:
\begin{equation}
\label{eq:sec_ham}
\left< h_1 \right> = 
\frac{1}{(2\pi)^2}
\int_0^{2\pi} \! \int_0^{2\pi} h_1 \, d\lambda d\lambda' .
\end{equation}

As both the direct and the indirect parts of $h_1$ are expressed in terms of the excentric longitudes $F$ and $F'$, we use $d\lambda/dF = \rho = 1 - \vecrho_0 \cdot \vecz$ to change the integration variables:
\begin{equation}
\label{eq:sec_ham_FF'}
\left< h_1 \right> = 
\frac{1}{(2\pi)^2}
\int_0^{2\pi} \! \int_0^{2\pi} \rho \rho' \, h_1 \, dF dF' .
\end{equation}
Since $\left< \vec{p} \right> = \vec{0}$, the indirect part of the disturbing function does not contribute to the integral. To compute the average of the direct part, we recall that in Eq.~\eqref{eq:legendre_trunc} the term $\Vert\vec\nu\Vert$ only depends on $F'-F$. Therefore, we consider the additional change of variables $F' \rightarrow \widetilde{F} = F'-F$ to obtain:
\begin{equation}
\label{eq:sec_direct_part}
\left< \frac{a'}{\Delta} \right> = 
\frac{1}{(2\pi)^2}
\int_0^{2\pi} \! \int_0^{2\pi} \rho \rho' \frac{a'}{\Delta} \, dF d \widetilde{F} .
\end{equation}
To explicitly compute the average of Eq.~\eqref{eq:legendre_trunc}, we first notice that: 
\begin{equation}
\label{eq:vecrho0_prime}
\vecrho_0' = \cos(\widetilde{F}) \, \vecrho_0 + \sin(\widetilde{F}) \, \vechat{k} \times \vecrho_0 .
\end{equation}
Using this relation, one can express all the scalar products involving vector $\vecrho_0'$ as:
\begin{equation}
\label{eq:scalar_products_rho'}
\vec\rho_0' \cdot \vec{u}_i = \cos(\widetilde{F}) \, \vecrho_0 \cdot \vec{u}_i + \sin(\widetilde{F}) \, \vecrho_0 \cdot (\vec{u}_i \times \vechat{k}) ,
\end{equation}
where $\vec{u}_i \in V_1$. 
Replacing this identity in Eq.~\eqref{eq:legendre_trunc}, the eight different scalar products of degree 1 in eccentricity and inclinations that appear in the direct part can all be written in the following form:
\begin{equation}
\label{eq:terms_degree_1_v2}
\vec\rho_0 \cdot \vec{u}_i, \quad \mathrm{with}\enspace  
\vec{u}_i \in V_1 \cup \widetilde{V}_1, \ \enspace 
\widetilde{V}_1 = \{ 
\vecz \times \vechat{k}, \vecz' \times \vechat{k}, \veczeta \times \vechat{k}, \veczeta' \times \vechat{k} \} .
\end{equation}
We notice that all the vectors in $\widetilde{V}_1$ are parallel to the invariant plane, similarly to those in the set $V_1$.

Once the direct part has been expressed as a function of $F$ and $\widetilde{F}$, the integral of Eq.~\eqref{eq:legendre_trunc} over the variable $\widetilde{F}$ can be readily computed in terms of Laplace coefficients \cite[][and references therein]{Laskar1995}:
\begin{equation}
\label{eq:laplace_coeffs}
\frac{1}{2\pi} \int_0^{2\pi} \frac{\E^{-\iota \kappa \widetilde{F}}}{\Vert\vec\nu\Vert^s} d \widetilde{F} = 
\frac{1}{2} \lapl{s}{\kappa}(\alpha) ,
\end{equation}
where $\kappa \in \mathbb{Z}$. 

We are then left with a semi-averaged direct part that only depends on the eccentric longitude $F$ through the scalar products in Eq.~\eqref{eq:terms_degree_1_v2}. Therefore, the integral over $F$ appearing in Eq.~\eqref{eq:sec_direct_part} involves averages of the following form:
\begin{equation}
\left< \prod_{i=1}^n \vec\rho_0 \cdot \vec{u}_i \right> ,
\end{equation}
where $\vec{u}_i \in V_1 \cup \widetilde{V}_1$ (two or more of these vectors can be identical) and $n \leq d$ if the direct part is truncated to total degree $d$ in eccentricity and inclination. For $n$ small, these averages can be easily computed: 
\begin{equation}
\label{eq:averages}
\begin{aligned}
\left< \vec\rho_0 \cdot \vec{u}_1 \right> &= 0, \\ 
\left< (\vec\rho_0 \cdot \vec{u}_1)(\vec\rho_0 \cdot \vec{u}_2) \right> &= \frac{\vec{u}_1 \cdot \vec{u}_2}{2}, \\
\left< (\vec\rho_0 \cdot \vec{u}_1)(\vec\rho_0 \cdot \vec{u}_2)(\vec\rho_0 \cdot \vec{u}_3) \right> &= 0, \\
\left< (\vec\rho_0 \cdot \vec{u}_1)(\vec\rho_0 \cdot \vec{u}_2)(\vec\rho_0 \cdot \vec{u}_3)(\vec\rho_0 \cdot \vec{u}_4) \right> 
&= \frac{1}{8} \left[ 
(\vec{u}_1 \cdot \vec{u}_2)(\vec{u}_3 \cdot \vec{u}_4) \right. \\
 &+ (\vec{u}_1 \cdot \vec{u}_3)(\vec{u}_2 \cdot \vec{u}_4) + \left.(\vec{u}_1 \cdot \vec{u}_4)(\vec{u}_2 \cdot \vec{u}_3) \right] ,
\end{aligned}
\end{equation}
for all $\vec{u}_1,\vec{u}_2,\vec{u}_3,\vec{u}_4 \in V_1 \cup \widetilde{V}_1$. 
The averages are zero for all odd values of $n$, in agreement with d'Alembert rules \cite[e.g.,][]{Morbidelli2002}. For bigger even values of $n$ the expressions become more involved. We present here a proposition that gives the general form of these averages.
\begin{proposition}[]
\label{proposition1}
Let $\vec{u}_1, \vec{u}_2, \dots, \vec{u}_n \in \mathbb{R}^3$ be such that $\vec{u}_i \cdot \vechat{k} = 0$ for $1 \leq i \leq n$. Then:
\begin{equation}
\label{eq:proposition1}
\left< \prod_{i=1}^n \vec\rho_0 \cdot \vec{u}_i \right> = 
\begin{dcases}
    0 \vphantom{\frac{0}{0}} & \text{for $n$ odd}\\
    \frac{1}{2^{n/2} (n/2)!} \sum_{\sigma \in \widetilde{\mathfrak{S}}_{n}} \left[ \prod_{i=1}^{n/2} \vec{u}_{\sigma(2i-1)} \cdot \vec{u}_{\sigma(2i)} \right]  & \text{for $n$ even}
\end{dcases} ,
\end{equation}
where $\mathfrak{S}_n$ is the symmetric group of all permutations of the set $\{1, 2, \dots, n\}$ and $\widetilde{\mathfrak{S}}_n = \{ \sigma \in \mathfrak{S}_n : \sigma(2i-1) < \sigma(2i+1) \ \mathrm{for} \ i = 1,2,\dots,n/2-1 \ \mathrm{and} \ \sigma(2i-1) < \sigma(2i) \ \mathrm{for} \ i = 1,2,\dots,n/2 \}$. 
\end{proposition}
\begin{proof}
By assumption, the vectors $\vec{u}_1, \vec{u}_2, \dots, \vec{u}_n$ can be written as follows:
\begin{equation}
\vec{u}_i = U_i \begin{bmatrix} \cos{\theta_i} \\ \sin{\theta_i} \\ 0 \end{bmatrix}_\mathcal{I} \quad \text{for $1 \leq i \leq n$} . 
\end{equation}
From Eq.~\eqref{eq:rho0}, one then has $\prod_{i=1}^n \vec\rho_0 \cdot \vec{u}_i = \prod_{i=1}^n U_i \cos(F-\theta_i)$. 
We then consider the following identity:
\begin{equation}
\prod_{i=1}^n \cos(\alpha_i) = \frac{1}{2^n} \sum_{\vec{l} \in L} \cos(\vec{l} \cdot \vec{\alpha}) ,
\end{equation}
where $L = \{-1,1\}^n$ and $\vec\alpha = (\alpha_1, \alpha_2, \dots, \alpha_n)$. Using this identity, we readily obtain: 
\begin{equation}
\prod_{i=1}^n \cos(F-\theta_i) = \frac{1}{2^n} \sum_{\vec{l} \in L} \cos \left[ \left( \sum_{i=1}^n l_i \right) F - \vec{l} \cdot \vec{\theta} \right] ,
\end{equation}
where $\vec{l} = (l_1, l_2, \dots, l_n)$ and $\vec\theta = (\theta_1, \theta_2, \dots, \theta_n)$. Therefore, the secular average reads as:
\begin{equation}
\left< \prod_{i=1}^n \cos(F-\theta_i) \right> = \frac{1}{2^n} \sum_{\substack{\vec{l} \in L \\ \sum_{i=1}^n l_i = 0}} \cos(\vec{l} \cdot \vec{\theta}) \ .
\end{equation}
From the constrain on the last summation, it straightforwardly follows that the average is zero for $n$ odd. Therefore, in the following we assume that $n = 2m$ is a positive and even integer. In this case, we can write the summation in the last equation in terms of permutations: 
\begin{equation}
\sum_{\substack{\vec{l} \in L \\ \sum_{i=1}^{2m} l_i = 0}} \cos(\vec{l} \cdot \vec{\theta}) =  
\frac{1}{(m!)^2}
\sum_{\sigma \in \mathfrak{S}_{2m}} \cos \left[ (\theta_{\sigma(1)} - \theta_{\sigma(2)}) + \dots + (\theta_{\sigma(2m-1)} - \theta_{\sigma(2m)}) \right] \ .
\end{equation}
Using the identity $\cos(\alpha+\beta) = \cos(\alpha)\cos(\beta) - \sin(\alpha)\sin(\beta)$ iteratively, and taking into account the change of sign of the sine function under permutation, we get:
\begin{equation}
\begin{aligned}
&\sum_{\sigma \in \mathfrak{S}_{2m}} \cos \left[ (\theta_{\sigma(1)} - \theta_{\sigma(2)}) + \dots + (\theta_{\sigma(2m-1)} - \theta_{\sigma(2m)}) \right] = \\
&\sum_{\sigma \in \mathfrak{S}_{2m}} \cos(\theta_{\sigma(1)} - \theta_{\sigma(2)}) \cos \left[ (\theta_{\sigma(3)} - \theta_{\sigma(4)}) + \dots  + (\theta_{\sigma(2m-1)} - \theta_{\sigma(2m)})\right] = \\
&\sum_{\sigma \in \mathfrak{S}_{2m}} \left[ \prod_{i=1}^{m} \cos(\theta_{\sigma(2i-1)} - \theta_{\sigma(2i)})  \right] .
\end{aligned}
\end{equation}
The summation over the symmetric group can be further simplified: 
\begin{equation}
\sum_{\sigma \in \mathfrak{S}_{2m}} \left[ \prod_{i=1}^{m} \cos(\theta_{\sigma(2i-1)} - \theta_{\sigma(2i)})  \right] = 
2^m m! \sum_{\sigma \in \widetilde{\mathfrak{S}}_{2m}} \left[ \prod_{i=1}^{m} \cos(\theta_{\sigma(2i-1)} - \theta_{\sigma(2i)})  \right] ,
\end{equation}
where last summation is restricted to permutations $\sigma$ such that $\sigma(2i-1) < \sigma(2i+1) \ \forall i \in \{1,2,\dots,m-1\}$ and $\sigma(2i-1) < \sigma(2i) \ \forall i \in \{1,2,\dots,m\}$. We notice that, for such permutations, one always has $\sigma(1) = 1$. Aggregating all previous identities, one obtains:
\begin{equation}
\begin{aligned}
\prod_{i=1}^n \left< \vec\rho_0 \cdot \vec{u}_i \right> &= \prod_{i=1}^n U_i \left< \cos(F-\theta_i) \right> \\
&= \frac{1}{2^m m!} \sum_{\sigma \in \widetilde{\mathfrak{S}}_{2m}} \left[ \prod_{i=1}^{m} U_{\sigma(2i-1)} U_{\sigma(2i)} \cos(\theta_{\sigma(2i-1)} - \theta_{\sigma(2i)}) \right] \\
&= \frac{1}{2^m m!} \sum_{\sigma \in \widetilde{\mathfrak{S}}_{2m}} \left[ \prod_{i=1}^{m} \vec{u}_{\sigma(2i-1)} \cdot \vec{u}_{\sigma(2i)} \right] .
\end{aligned}
\end{equation}
\end{proof}
When implemented in a computer algebra system, Eq.~\eqref{eq:proposition1} allows us to systematically perform the average over the eccentric longitude $F$ at any degree in eccentricities and inclinations. 

Equation~\eqref{eq:proposition1} implies that in the secular disturbing function resulting from Eq.~\eqref{eq:sec_direct_part} there appear the following three different kinds of scalar product:
\begin{equation}
\label{eq:scalar_products_with_k}
\begin{aligned}
&\vec{u}_i \cdot \vec{u}_j , \\
&\vec{u}_i \cdot (\vec{u}_j \times \vechat{k}) = \vechat{k} \cdot (\vec{u}_i \times \vec{u}_j) , \\
&(\vec{u}_i \times \vechat{k}) \cdot (\vec{u}_j \times \vechat{k}) = 
(\vec{u}_i \cdot \vec{u}_j)(\vechat{k} \cdot \vechat{k})
- (\vec{u}_i \cdot \vechat{k})(\vec{u}_j \cdot \vechat{k}) = \vec{u}_i \cdot \vec{u}_j ,
\end{aligned}
\end{equation}
where $\vec{u}_i, \vec{u}_j \in V_1$, and we used the identities $\vec{u}_i \cdot \vechat{k} = 0$ and $\vechat{k} \cdot \vechat{k} = 1$. Moreover, from Eq.~\eqref{eq:scalar_products_rho'} and the fact that $\int (\cos\widetilde{F})^n (\sin\widetilde{F})^m / \Vert\vec\nu\Vert^s \, d \widetilde{F} = 0$ for all $m$ odd integers, it follows that the second type of terms in Eq.~\eqref{eq:scalar_products_with_k} only enters the secular disturbing function through the following expressions:
\begin{equation}
\begin{aligned}
[ \vechat{k} \cdot (\vec{u}_i \times \vec{u}_j) ] [ \vechat{k} \cdot (\vec{u}_\kappa \times \vec{u}_l) ] &= 
(\vec{u}_i \times \vec{u}_j) \cdot (\vec{u}_\kappa \times \vec{u}_l) \\
&= (\vec{u}_i \cdot \vec{u}_\kappa)(\vec{u}_j \cdot \vec{u}_l) - (\vec{u}_i \cdot \vec{u}_l)(\vec{u}_j \cdot \vec{u}_\kappa) ,
\end{aligned}
\end{equation}
where $\vec{u}_i, \vec{u}_j, \vec{u}_\kappa, \vec{u}_l \in V_1$ and the first identity comes from the fact that $\vec{u}_i \times \vec{u}_j$ is perpendicular to the invariant plane, and thus aligned with $\vechat{k}$. 

From all the previous considerations, it follows that the secular disturbing function can be fully expressed in terms of scalar products of vectors belonging to the set $V_1$. Moreover, in line with the simplifications demonstrated in \citep{Laskar1995}, the resulting expansion can be further simplified by expressing the coefficients of all the terms of the same degree $d$ by means of two algebraically independent Laplace coefficients:
\begin{equation}
\begin{aligned}
    b_{(d+1)/2}^{(0)}, \ b_{(d+1)/2}^{(1)} \quad \text{for $d$ even} , \\
    b_{[d/2]+3/2}^{(0)}, \ b_{[d/2]+3/2}^{(1)} \quad \text{for $d$ odd} .
\end{aligned}
\end{equation}

As an example, we give in Appendix~\ref{sec:A1_secular_ham_4} the expression of the secular disturbing function up to degree 4 in eccentricities and inclinations, which has been obtained by means of the computer algebra system TRIP.

\section{Arbitrary inequality of the disturbing function}
\label{sect:sec_inequalities}
The disturbing function $h_1$ can be decomposed as a Fourier series of the mean longitudes $\lambda$ and $\lambda'$:
\begin{equation}
\label{eq:fourier}
h_1 = \sum_{\kappa,\kappa' \in \mathbb{Z}} C_{\kappa,\kappa'} \E^{\iota (\kappa \lambda + \kappa' \lambda')} ,
\end{equation}
where the Fourier coefficient corresponding to the inequality $\kappa,\kappa'$ is given by
\begin{equation}
\label{eq:fourier_coeffs}
C_{\kappa,\kappa'} = \left< h_1 \E^{-\iota (\kappa \lambda + \kappa' \lambda')} \right> = 
\frac{1}{(2\pi)^2} \int_0^{2\pi} \! \int_0^{2\pi} \rho \rho' h_1 \E^{-\iota (\kappa \lambda + \kappa' \lambda')} \, dF d \widetilde{F} .
\end{equation}
We illustrate here how to compute any of the inequalities of the disturbing function, generalising the secular average discussed in the previous section. 

In order to obtain final expressions that are real and in a vector form, we exploit the fact that the planetary Hamiltonian is a real function by writing: 
\begin{equation}
\begin{aligned}
h_1 &= \sum_{\substack{\kappa,\kappa' \in \mathbb{Z} \\ \kappa + \kappa' \geq 0}} h_{1}^{\kappa,\kappa'} , \\
h_{1}^{\kappa,\kappa'} &= 
\begin{dcases}
    C_{0,0} & \text{for $\kappa = 0, \kappa' = 0$} \\
    2 \real \left( C_{\kappa,\kappa'} \E^{\iota (\kappa \lambda + \kappa' \lambda')} \right) & \text{otherwise}
\end{dcases} .
\end{aligned}
\end{equation}
To perform the integration in Eq.~\eqref{eq:fourier_coeffs}, we first express the exponential factor in terms of eccentric longitudes using Kepler's equation~\eqref{eq:kepler}. 
From the classical Jacobi–Anger expansions \citep{NIST:DLMF}, one gets:
\begin{equation}
\E^{\iota \kappa \imag(z \E^{-\iota F})} = \sum_{n=-\infty}^{+\infty} (-1)^n J_n(\kappa e) \E^{-\iota n \varpi} \E^{\iota n F} ,
\end{equation}
where $J_n$ are Bessel functions of the first kind. Then, expressing Bessel functions as power series, we write:
\begin{equation}
J_n(\kappa e) \E^{-\iota n \varpi} = {\left( \frac{\kappa \bar{z}}{2} \right)}^n f_n \! \left( \frac{\kappa^2 z \bar{z}}{4} \right) \enspace \mathrm{for} \ n\geq0 , 
\end{equation}
where we defined the functions $f_n$ as the following power series: 
\begin{equation}
\label{eq:f_n}
f_n(t) = \sum_{l=0}^{+\infty} \frac{(-1)^l}{l! (n+l)!} t^l .
\end{equation}
Using Eq.~\eqref{eq:kepler}, we finally get:
\begin{equation}
\begin{aligned}
\label{eq:expi_lambda}
\E^{\iota \kappa \lambda} &= \eta_{\kappa} \E^{\iota \kappa F} , \\
\eta_\kappa &= f_0 \! \left( \frac{\kappa^2 z \bar{z}}{4} \right) + \sum_{n=1}^{+\infty} \kappa^n \frac{ (-\bar{z})^n \E^{\iota n F} + z^n \E^{-\iota n F}}{2^n} f_n \! \left( \frac{\kappa^2 z \bar{z}}{4} \right) .
\end{aligned}
\end{equation}

At this point, one can write $\eta_\kappa$ as a function of the vectors $\vecrho_0,\vecz$ by observing that Eq.~\eqref{eq:vectors_complex_vars} gives:
\begin{equation}
\label{eq:complex_to_vecs}
u_i \bar{u}_j = \vec{u}_i \cdot \vec{u}_j - \iota \vechat{k} \cdot (\vec{u}_i \times \vec{u}_j) ,
\end{equation}
where $\vec{u}_i, \vec{u}_j \in V$ and $u_i, u_j$ are their complex counterparts. By defining the functions
\begin{equation}
\begin{aligned}
g_n(\vec{u}_i,\vec{u}_j) &= \left( \vec{u}_i \cdot \vec{u}_j + \iota \vec{u}_i \cdot (\vec{u}_j \times \vechat{k}) \right)^n \\
&= \sum_{l=0}^n (-1)^{\lfloor l/2 \rfloor} \binom{n}{l}
\left( \iota \vec{u}_i \cdot (\vec{u}_j \times \vechat{k}) \right)^{l \bmod{2}} 
\Vert \vec{u}_i \times \vec{u}_j \Vert^{2 \lfloor l/2 \rfloor} 
(\vec{u}_i \cdot \vec{u}_j)^{n-l} ,
\end{aligned}
\end{equation}
we can restate $\eta_\kappa$ as follows:
\begin{equation}
\begin{aligned}
\eta_\kappa &= \eta_\kappa(\vecrho_0,\vecz) \\
&= f_0 \! \left( \frac{\kappa^2 \Vert\vecz\Vert^2}{4} \right) 
+ \sum_{n=1}^{+\infty} \kappa^n \frac{(-1)^n \overline{g_n(\vecrho_0,\vecz)} + g_n(\vecrho_0,\vecz)}{2^n} \, f_n \! \left( \frac{\kappa^2 \Vert\vecz\Vert^2}{4} \right) ,
\end{aligned}
\end{equation}
As a result, $\eta_\kappa$ can be fully expressed in terms of scalar products of the form~\eqref{eq:terms_degree_1_v2}. We can then restate Eq.~\eqref{eq:fourier_coeffs} as follows:
\begin{equation}
\label{eq:fourier_coeffs_2}
C_{\kappa,\kappa'} = \left< h_1 \left( \overline{\eta}_{\kappa'}' \E^{-\iota \kappa' \widetilde{F}} \right) \overline{\eta}_\kappa \E^{-\iota (\kappa + \kappa') F} \right> ,
\end{equation}
where $\overline{\eta}_{\kappa}' \vcentcolon= \overline{\eta_{\kappa}(\vecrho_0',\vecz')}$ and we point out that $g_n(\vecrho_0',\vecz') = \E^{-\iota n \widetilde{F}} g_n(\vecrho_0,\vecz')$. 

Now, we observe that Eq.~\eqref{eq:complex_to_vecs} allows us to write: 
\begin{equation}
\label{eq:exp_factors}
\begin{aligned}
\E^{-\iota (\kappa + \kappa') F} \E^{\iota (\kappa + \kappa') \lambda} &= 
g_{\kappa+\kappa'}(\vecrho_0,\vec\ell) \\
\E^{\iota \kappa' (\lambda'-\lambda)} &= 
g_{\kappa'}(\vec\ell,\vec\ell')
\end{aligned} ,
\end{equation}
where we defined the vector $\vec\ell$ as
\begin{equation}
\label{eq:vec_ell}
\vec\ell = \begin{bmatrix} \cos \lambda \\ \sin \lambda \\ 0 \end{bmatrix}_\mathcal{I} ,
\end{equation}
and assumed both $\kappa + \kappa' \geq 0$ and $\kappa' \geq 0$\footnote{Generalization to the other cases is straightforward through complex conjugation in Eq.~\eqref{eq:exp_factors}.}. We point out that, similarly to $\vecrho_0,\vecz$ and $\veczeta$, the vector $\vec\ell$ also lies in the invariant plane of the planetary system. Employing the expressions~\eqref{eq:exp_factors}, one can write:
\begin{equation}
\label{eq:fourier_coeffs_3}
C_{\kappa,\kappa'} \E^{\iota (\kappa \lambda + \kappa' \lambda')} 
= \left< 
h_1 \left( \overline{\eta}_{\kappa'}' \E^{-\iota \kappa' \widetilde{F}} \right) \overline{\eta}_\kappa 
\, g_{\kappa+\kappa'}(\vecrho_0,\vec\ell)
\right> 
g_{\kappa'}(\vec\ell,\vec\ell') ,
\end{equation}
where $\vec\ell$ is meant to be a fixed parameter in the integration with respect to $F$. 

Using Eq.~\eqref{eq:f_n}, the function $\eta_\kappa$ can be straightforwardly expanded in series and truncated at a given maximum degree in eccentricity. Truncating the disturbing function $h_1$ as well, and expressing $\vecrho_0'$ in terms of $\vecrho_0$ and $\widetilde{F}$ via Eq.~\eqref{eq:vecrho0_prime}, the integral over $\widetilde{F}$ in Eq.~\eqref{eq:fourier_coeffs_3} can be computed in terms of Laplace coefficients. The remaining integral over $F$ can be then performed using Proposition~\ref{proposition1}, by applying Eq.~\eqref{eq:proposition1} with $\vec{u}_i \in V_1 \cup \widetilde{V}_1 \cup \{\vec\ell, \vec\ell \times \vechat{k}\}$. As a result, any inequality $h_{1}^{\kappa,\kappa'}$ of the disturbing function can be fully expressed in terms of scalar products of vectors belonging to the set $V_1 \cup \{ \vec\ell, \vec\ell' \}$.
We stress that the algorithm presented here allows for the direct computation of any specific inequality, without the need to determine the entire Fourier expansion of the disturbing function. This significantly reduces the number of terms to be stored and processed, which is particularly advantageous when constructing simplified dynamical models, such as those related to mean-motion resonances. 

We give in Appendix~\ref{sec:A1_inequalities} the expression of the inequalities of the direct part of the disturbing function related to the mean-motion resonances 1:1, 2:1, and 3:2, expanded up to degree 2 in eccentricities and inclinations, which has been obtained by means of the computer algebra system TRIP. 

As a final remark, we point out that an obvious alternative way of computing the integral~\eqref{eq:fourier_coeffs} is to express the disturbing function $h_1$ in terms of the complex variables $\E^{\iota F}, \E^{\iota F'}, z, z', \zeta, \zeta'$. Switching from vectors to complex variables is easily achieved considering Eq.~\eqref{eq:complex_to_vecs}, which gives: 
\begin{equation}
\vec{u}_i \cdot \vec{u}_j = \real(u_i \bar{u}_j) .
\end{equation}
This allows us to express all the scalar products appearing in the disturbing function in terms of complex variables. Once this is done, a finite Fourier expansion of $h_1 \E^{-\iota (k \lambda + k' \lambda')}$ with respect to $F$ and $F'$ is obtained through Eq.~\eqref{eq:expi_lambda} for a given maximum degree in eccentricities and inclinations. Any Fourier harmonic can be then readily extracted when employing a computer algebra system. Clearly, this procedure can also be applied as an alternative to the algorithm presented in Sect.~\ref{sect:sec_ham} to compute the secular disturbing function, that is, the coefficient $C_{0,0}$, in terms of complex variables. 

\section{Expressions involving angular momentum and eccentricity vectors}
\label{sec:eh}
We illustrate now how to express the inequalities of the disturbing function we obtained in the previous sections, in terms of the dimensionless angular momentum and eccentricity vectors of each orbit, that is, $\vech$ and $\vece$. We will retrieve in particular \citet{Abdullah2001a,Abdullah2001}'s expansion of the secular Hamiltonian.

The procedure is straightforward and based on inverting Eqs.~\eqref{eq:angular_momentum} and \eqref{eq:ecc_vec}, which allows us to write the vectors $\vecz, \veczeta$ as a function of $\vece, \vec{h}$:
\begin{equation}
\label{eq:z_zeta_to_e_h}
\begin{aligned}
\vecz &= \vece - \frac{\vece \cdot \vechat{k}}{2 \phi \chi^2} (\vec{h} + \phi \vechat{k}), \\
\veczeta &= \frac{ \vechat{k} \times \vec{h} }{2\phi\chi} ,
\end{aligned}
\end{equation}
with $\phi = \sqrt{1-\vece \cdot \vece}$ and $\phi \chi^2 = (\phi+\vec{h} \cdot \vechat{k})/2$.
By means of these identities, any expansion containing the vectors $\vecz,\veczeta,\vecz',\veczeta'$ can be expressed in terms of $\vece,\vech,\vece',\vech'$ and $\vechat{k}$. We recall that $\vechat{k}$ is the vertical axis of the inertial reference frame $\mathcal{I}$ in use, and is supposed to be aligned with the total angular momentum of the planetary system. 


It is important to remark that the ordering of the series terms is altered by Eq.~\eqref{eq:z_zeta_to_e_h}, which mixes together terms of different degree in eccentricities and inclinations. The correct ordering can be restored by noting that all the following quantities are of degree 2 in eccentricities and inclinations\footnote{We recall that $\vec{e} \cdot \vec{h} = 0$ and $\vech \cdot \vech = 1 - \vece \cdot \vece$.}:
\begin{equation}
\begin{aligned}
\vece \cdot \vece , \quad \vece \cdot \vece', \quad \vece \cdot \vech', \quad 1 - \vech \cdot \vech' ,& \\
\vece \cdot \vechat{k} , \quad 1 - \vech \cdot \vechat{k},
\end{aligned}
\end{equation}
while the following ones, appearing in non-secular inequalities, are of degree 1:
\begin{equation}
\vec\ell \cdot \vece , \quad \vec\ell \cdot \vece' , \quad \vec\ell \cdot \vech , \quad \vec\ell \cdot \vech' .
\end{equation}

As an example of the resulting expansions, we first report the secular disturbing function truncated to degree 4 in eccentricities and inclinations.
\begingroup
\allowdisplaybreaks
\begin{align*}
\label{eq:secular_eh}
\left< \frac{a'}{\Delta} \right> =  \stepcounter{equation}\tag{\theequation} \\
            \frac{1}{2}         \lapl{1}{0} \\
 
 -          \frac{1}{4}         \lapl{3}{1} \alpha \left( 1 - \vech \cdot \vech' \right) \\
 +          \frac{1}{4}         \lapl{3}{1} \alpha \left( \vece \cdot \vece \right) \\
 + \left(  - \frac{1}{2} \lapl{3}{1} + \frac{3}{4} \alpha \lapl{3}{0} - \frac{1}{2} \alpha^{2} \lapl{3}{1} \right) \left( \vece \cdot \vece' \right) \\
 +          \frac{1}{4}         \lapl{3}{1} \alpha \left( \vece' \cdot \vece' \right) \\

+ \left(  - \frac{3}{16} \alpha \lapl{5}{1} + \frac{21}{32} \alpha^{2} \lapl{5}{0} - \frac{3}{16} \alpha^{3} \lapl{5}{1} \right) \left( 1 - \vech \cdot \vech' \right)^{2} \\
 +          \frac{9}{32}        \lapl{5}{1} \alpha^{3} \left( \vece \cdot \vece \right)^{2} \\
 + \left(  - \frac{9}{16} \alpha \lapl{5}{1} + \frac{45}{32} \alpha^{2} \lapl{5}{0} - \frac{9}{16} \alpha^{3} \lapl{5}{1} \right) \left( \vece \cdot \vece' \right)^{2} \\
 + \left(  - \frac{9}{16} \alpha \lapl{5}{1} + \frac{15}{32} \alpha^{2} \lapl{5}{0} - \frac{3}{16} \alpha^{3} \lapl{5}{1} \right) \left( \vece \cdot \vech' \right)^{2} \\
 +          \frac{9}{32}        \lapl{5}{1} \alpha \left( \vece' \cdot \vece' \right)^{2} \\
 + \left(  - \frac{3}{16} \alpha \lapl{5}{1} + \frac{15}{32} \alpha^{2} \lapl{5}{0} - \frac{9}{16} \alpha^{3} \lapl{5}{1} \right) \left( \vece' \cdot \vech \right)^{2} \\
 + \left(    \frac{3}{8} \alpha \lapl{5}{1} - \frac{15}{16} \alpha^{2} \lapl{5}{0} \right) \left( 1 - \vech \cdot \vech' \right) \left( \vece \cdot \vece \right) \\
 + \left(  - \frac{3}{4} \lapl{5}{1} + \frac{15}{8} \alpha \lapl{5}{0} - \frac{27}{16} \alpha^{2} \lapl{5}{1} + \frac{15}{8} \alpha^{3} \lapl{5}{0} - \frac{3}{4} \alpha^{4} \lapl{5}{1} \right) \left( 1 - \vech \cdot \vech' \right) \left( \vece \cdot \vece' \right) \\
 + \left(  - \frac{15}{16} \alpha^{2} \lapl{5}{0} + \frac{3}{8} \alpha^{3} \lapl{5}{1} \right) \left( 1 - \vech \cdot \vech' \right) \left( \vece' \cdot \vece' \right) \\
 + \left(  - \frac{15}{16} \alpha^{3} \lapl{5}{0} + \frac{3}{8} \alpha^{4} \lapl{5}{1} \right) \left( \vece \cdot \vece \right) \left( \vece \cdot \vece' \right) \\
 +          \frac{9}{16}        \lapl{5}{0} \alpha^{2} \left( \vece \cdot \vece \right) \left( \vece' \cdot \vece' \right) \\
 + \left(    \frac{3}{8} \lapl{5}{1} - \frac{15}{16} \alpha \lapl{5}{0} \right) \left( \vece \cdot \vece' \right) \left( \vece' \cdot \vece' \right) \\
 + \left(    \frac{3}{4} \lapl{5}{1} - \frac{15}{8} \alpha \lapl{5}{0} + \frac{27}{16} \alpha^{2} \lapl{5}{1} - \frac{15}{8} \alpha^{3} \lapl{5}{0} + \frac{3}{4} \alpha^{4} \lapl{5}{1} \right) \left( \vece \cdot \vech' \right) \left( \vece' \cdot \vech \right) \\

\end{align*}
\endgroup
The expression coincides with Abdullah's expansion in the form given in \cite[Eq.~(26)]{Boue2014}. We note that the vector $\vechat{k}$ does not appear in Eq.~\eqref{eq:secular_eh}. Therefore, its precise definition is not important as long as it is chosen to be close to the total angular momentum vector of the planetary system. Choices such as $\vechat{k} = \vechat{h}$ or $\vechat{k} = \vechat{h}'$, for example, would have led to the same result. 

As an example of non-secular inequality, we consider the Fourier harmonic of the direct part of the disturbing function that is related to the mean-motion resonance~1:1. We compute the quantity $2 \real \left( C_{-1,1}^\mathrm{dir} \E^{\iota (- \lambda + \lambda')} \right)$, where 
\begin{equation}
C_{\kappa,\kappa'}^\mathrm{dir} = \left< \frac{a'}{\Delta} \, \E^{-\iota (\kappa \lambda + \kappa' \lambda')} \right> ,
\end{equation}
and truncate the expansion to degree 2 in eccentricities and inclinations. 
\begingroup
\allowdisplaybreaks
\begin{align*}
\label{eq:MMR_11_eh}
2 \real \left( C_{-1,1}^\mathrm{dir} \E^{\iota (- \lambda + \lambda')} \right) = \stepcounter{equation}\tag{\theequation} \\
                       \lapl{1}{1} \left( \vec\ell \cdot \vec\ell' \right) \\
 + \left(  - \frac{1}{2} \lapl{3}{1} + \frac{1}{2} \alpha \lapl{3}{0} - \frac{1}{2} \alpha^{2} \lapl{3}{1} \right) \left( 1 - \vech \cdot \vech' \right) \left( \vec\ell \cdot \vec\ell' \right) \\
 + \left(    \frac{5}{4} \lapl{3}{1} - 2 \alpha \lapl{3}{0} + \frac{5}{4} \alpha^{2} \lapl{3}{1} \right) \left( \vece \cdot \vece \right) \left( \vec\ell \cdot \vec\ell' \right) \\
 + \left(  - \frac{5}{3} \alpha^{-1} \lapl{3}{1} + \frac{5}{2} \lapl{3}{0} - \frac{13}{6} \alpha \lapl{3}{1} + \frac{5}{2} \alpha^{2} \lapl{3}{0} - \frac{5}{3} \alpha^{3} \lapl{3}{1} \right) \left( \vece \cdot \vece' \right) \left( \vec\ell \cdot \vec\ell' \right) \\
 + \left(    \frac{5}{3} \alpha^{-1} \lapl{3}{1} - \frac{5}{2} \lapl{3}{0} + \frac{5}{3} \alpha \lapl{3}{1} - \frac{5}{2} \alpha^{2} \lapl{3}{0} + \frac{5}{3} \alpha^{3} \lapl{3}{1} \right) \left( \vece \cdot \vec\ell \right) \left( \vece' \cdot \vec\ell' \right) \\
 + \left(  - \frac{5}{3} \alpha^{-1} \lapl{3}{1} + \frac{5}{2} \lapl{3}{0} - \frac{5}{3} \alpha \lapl{3}{1} + \frac{5}{2} \alpha^{2} \lapl{3}{0} - \frac{5}{3} \alpha^{3} \lapl{3}{1} \right) \left( \vece \cdot \vec\ell' \right) \left( \vece' \cdot \vec\ell \right) \\
 + \left(    \frac{5}{4} \lapl{3}{1} - 2 \alpha \lapl{3}{0} + \frac{5}{4} \alpha^{2} \lapl{3}{1} \right) \left( \vece' \cdot \vece' \right) \left( \vec\ell \cdot \vec\ell' \right) \\
 + \left(  - \frac{1}{2} \lapl{3}{1} +  \alpha \lapl{3}{0} - \frac{1}{2} \alpha^{2} \lapl{3}{1} \right) \left( \vech \cdot \vec\ell \right) \left( \vech' \cdot \vec\ell' \right) \\
 + \left(    \frac{1}{2} \lapl{3}{1} -  \alpha \lapl{3}{0} + \frac{1}{2} \alpha^{2} \lapl{3}{1} \right) \left( \vech \cdot \vec\ell' \right) \left( \vech' \cdot \vec\ell \right) \\
 
\end{align*}
\endgroup
Here again the vector $\vechat{k}$ does not appear explicitly in the final expression. However, the vectors $\vec\ell, \vec\ell'$ do implicitly depend on $\vechat{k}$, as they are perpendicular to it by definition, and therefore lie in the reference plane of the frame $\mathcal{I}$ [see Eq.~\eqref{eq:vec_ell}]. This feature could be undesirable if one aims to describe the motion of a nearly-circular and quasi-coplanar planetary system that is not isolated, but interacts, for instance, with an inclined companion. In such a case, the total angular momentum of the planets is not constant, and as a consequence the reference frame $\mathcal{I}$ is not inertial. It is then possible to express the vectors $\vec\ell, \vec\ell'$ via an arbitrary inertial reference frame $\widetilde{\mathcal{I}}$ with a right-handed basis of unit vectors $\vechat{l}, \vechat{m}, \vechat{n}$. Using Eq.~\eqref{eq:vec_ell} and Lemma~\ref{lemma1}, we can write:
\begin{equation}
\label{eq:l_primed}
\begin{aligned}
\vec\ell &= \rot{\vec{k}}(\lambda) \vechat{\imath} = 
[\rot{\vec{\zeta}}(\inc)]^{-1} \rot{\vec{\zeta}}(\inc) \rot{\vec{k}}(\lambda) \vechat{\imath} =
R_{-\vechat{\zeta}}(\inc) \begin{bmatrix} \cos{M} \\ \sin{M} \\ 0 \end{bmatrix}_\mathcal{O} = 
R_{-\vechat{\zeta}}(\inc) R_{\vec{\tilde{\zeta}} / \Vert \vec{\tilde{\zeta}} \Vert}(\tilde{\inc}) \, \vec{\tilde{\ell}} , \\
&\mathrm{with} \quad \vec{\tilde{\ell}} = \begin{bmatrix} \cos{\tilde{\lambda}} \\ \sin{\tilde{\lambda}} \\ 0 \end{bmatrix}_{\widetilde{\mathcal{I}}}, 
\quad \vec{\tilde{\zeta}} = \frac{ \vechat{n} \times \vec{h} }{2\phi\widetilde{\chi}}, \quad \phi {\widetilde{\chi}}^2 = \frac{\phi+\vec{h} \cdot \vechat{n}}{2} . \\
\end{aligned}
\end{equation}
In these expressions, $\tilde\lambda$ and $\tilde\inc$ are the mean longitude and inclination of the unprimed orbit, respectively, with respect to the frame $\widetilde{\mathcal{I}}$, while $\vec{\tilde{\zeta}}$ points towards the corresponding ascending node. Analogous expressions hold for the primed orbit. Equation~\eqref{eq:l_primed} allows us to obtain the dependencies $\vec\ell = \vec\ell(\vec{\tilde\ell}, \vech,\vece; \vechat{k},\vechat{n})$ and $\vec\ell' = \vec\ell'(\vec{\tilde\ell}\,', \vech',\vece'; \vechat{k},\vechat{n})$. The most compact expressions are obtained with the choice $\vechat{k} = \vechat{h}$, which gives: 
\begin{equation}
\veczeta = \vec{0} , \qquad 
\veczeta' = \frac{ \vech \times \vech' }{2 \phi \phi' \chi'},
\end{equation}
where $\phi \phi' {\chi'}^2 = (\phi \phi' + \vech \cdot \vech')/2$. 

As a final remark, we observe that from Eqs.~\eqref{eq:rho_tilde_2} and \eqref{eq:z_zeta_to_e_h} it follows that:
\begin{equation}
\vecrhotilde = \vecrho_0 - \frac{\vecrho_0 \cdot \vech}{2 \phi^2 \chi^2} (\vech + \phi \vechat{k}) .
\end{equation}
Therefore, the expansion presented in Sec.~\ref{sect:expansion_direct_part} and subsequent ones can be carried out in terms of the angular momentum and eccentricity vectors from the beginning, without passing through the vectors $\vecz$ and $\veczeta$. 

\section{Discussion}
\label{sect:discussion}
In this work, we set up a vector-based formalism to develop the Hamiltonian of nearly-circular and quasi-coplanar planetary systems. The expansion is naturally expressed in terms of scalar products of vectors that lie in the system's invariant plane and encode the slow degrees of freedom of the dynamics ($\vecz, \veczeta$) as well as the fast ones ($\vecrho_0$ or $\vec\ell$). We demonstrate that the secular part of the disturbing function and all other non-secular inequalities (i.e., Fourier harmonics) can be developed in terms of such vectors. We also illustrate how to write the corresponding expansions in terms of the angular momentum and eccentricity vectors of the orbits. 

The usefulness of this kind of vector-based expansions lies in their obvious invariance under rotation. They can be easily evaluated in any reference system. As an example, this is of interest when a quasi-coplanar planetary system, interacting with an inclined companion, is entirely tilted to a significant degree with respect to some reference plane \cite{Boue2014}. In this regard, expansions involving the angular momentum and eccentricity vectors are the most useful ones. 

\backmatter





\bmhead{Acknowledgements}

We acknowledge the fruitful discussions with A.~Alnajjarine, J.~Laskar, and P.~Robutel and the kind hospitality of the Astronomy and Dynamical Systems team at the Laboratoire Temps Espace -- Paris Observatory during the preparation of part of this work. We also acknowledge the chance of employing the computer algebra system TRIP, which provides manipulation of symbolic vectors as required in the implementation of this work. 











\begin{appendices}

\section{Proof of Lemma \texorpdfstring{\ref{lemma1}}{1}}
\label{sec:A0}
We present here a demonstration of Eq.~\eqref{eq:lemma}, that is, $$\rot{\vec{h}}(\theta) \rot{\vec{\zeta}}(\inc) \vec{v} = \rot{\vec{\zeta}}(\inc) \rot{\vec{k}}(\theta) \vec{v} ,$$ for any $\vec{v} \in \mathbb{R}^3$ and $\theta \in [0,2\pi)$. The proof exploits the connection between composition of two rotations in $\mathbb{R}^3$ and multiplication of quaternions \cite{Cayley1845,Hamilton1853} (see \cite{Rodrigues1840,Gibbs1901} for an equivalent way of composing two spatial rotations). 

Consider the unit quaternions $q_1, q_2, q_3$ representing the rotations $\rot{\vec{h}}(\theta), \rot{\vec{\zeta}}(\inc), \rot{\vec{k}}(\theta)$, respectively. In the scalar-vector notation, one has:
\begin{equation}
\begin{aligned}
\rot{\vec{h}}(\theta) \quad \longleftrightarrow \quad q_1 &= \cos(\theta/2) + \sin(\theta/2) \vechat{h} \\
\rot{\vec{\zeta}}(\inc) \quad \longleftrightarrow \quad q_2 &= \chi + \veczeta \\
\rot{\vec{k}}(\theta) \quad \longleftrightarrow \quad q_3 &= \cos(\theta/2) + \sin(\theta/2) \vechat{k}
\end{aligned} ,
\end{equation}
where we used $\cos(I/2) = \sqrt{1-\veczeta\cdot\veczeta} = \chi$. 
The composite rotations $\rot{\vec{h}}(\theta) \rot{\vec{\zeta}}(\inc)$ and $\rot{\vec{\zeta}}(\inc) \rot{\vec{k}}(\theta)$ then correspond to the quaternion products $q_1q_2$ and $q_2q_3$, respectively:
\begin{equation}
\begin{aligned}
\rot{\vec{h}}(\theta) \rot{\vec{\zeta}}(\inc) \quad \longleftrightarrow \quad 
q_1 q_2 &= \cos(\theta/2) (\chi + \veczeta) + \sin(\theta/2) (\chi \vechat{h} + \vechat{h} \times \veczeta) \\
\rot{\vec{\zeta}}(\inc) \rot{\vec{k}}(\theta) \quad \longleftrightarrow \quad 
q_2 q_3 &= \cos(\theta/2) (\chi + \veczeta) + \sin(\theta/2) (\chi \vechat{k} + \veczeta \times \vechat{k})
\end{aligned} .
\end{equation}
where we used $\vechat{h} \cdot \veczeta = 0$ and $\vechat{k} \cdot \veczeta = 0$. By employing Eq.~\eqref{eq:angular_momentum}, one finds:
\begin{equation}
\begin{aligned}
\chi \vechat{h} + \vechat{h} \times \veczeta &= (1-2 \veczeta\cdot\veczeta) \chi \vechat{k} + 2 \chi^2 \veczeta \times \vechat{k} - (1-2 \veczeta\cdot\veczeta) \veczeta \times \vechat{k} + 2 \chi (\veczeta \cdot \veczeta) \vechat{k} \\
&= \chi \vechat{k} + \veczeta \times \vechat{k} .
\end{aligned}
\end{equation}
Therefore, one obtains $q_1q_2 = q_2q_3$ and $\rot{\vec{h}}(\theta) \rot{\vec{\zeta}}(\inc) = \rot{\vec{\zeta}}(\inc) \rot{\vec{k}}(\theta)$. 

\section{Examples of expansion}
\label{sec:A1}
This appendix presents examples that illustrate the expansions considered in this work. It is important to note that the listing of the series highlights the different scalar products of vectors involved in the development, rather than providing the most compact form of the expansion. 

The vectors $\vecrho_0, \vecz, \veczeta$ are defined in Eq.~\eqref{eq:vectors_complex_vars}.

\subsection{Direct part of the disturbing function to degree 2}
\label{sec:A1_direct_part_2}
We report here the expansion of the direct part of the disturbing function [Eq.~\eqref{eq:direct_part}] up to degree 2 in the eccentricities and inclinations of the orbits. 

We recall that $\Vert\vec\nu\Vert = \sqrt{1 + \alpha^2 -2\alpha \vec\rho_0 \cdot \vecrho_0'}$ and $\vec\rho_0 \cdot \vecrho_0' = \cos(F'-F)$.

\begin{align*}
\frac{a'}{\Delta} = \stepcounter{equation}\tag{\theequation} \\
                       \A{1} \\
 +                     \A{3} \alpha^2 \left(\vecrho_0 \cdot \vecz\right) \\
 -                     \A{3} \alpha \left(\vecrho_0 \cdot \vecz'\right) \\
 -                     \A{3} \alpha \left(\vecrho_0' \cdot \vecz\right) \\
 +                     \A{3} \left(\vecrho_0' \cdot \vecz'\right) \\
 + \left( - \frac{1}{2} \alpha^2 \A{3} + \frac{3}{2} \alpha^4 \A{5}\right) \left(\vecrho_0 \cdot \vecz\right)^2 \\
 +          \frac{3}{2}         \A{5} \alpha^2 \left(\vecrho_0 \cdot \vecz'\right)^2 \\
 +          \frac{3}{2}         \A{5} \alpha^2 \left(\vecrho_0' \cdot \vecz\right)^2 \\
 + \left(   \frac{3}{2} \A{5} - \frac{1}{2} \A{3}\right) \left(\vecrho_0' \cdot \vecz'\right)^2 \\
 +                     \A{3} \alpha \left(\vecz \cdot \vecz'\right) \\
 -          \frac{1}{2}         \A{3} \alpha \left(\vecrho_0 \cdot \vecrho_0'\right) \left(\vecz \cdot \vecz\right) \\
 -          \frac{1}{2}         \A{3} \alpha \left(\vecrho_0 \cdot \vecrho_0'\right) \left(\vecz' \cdot \vecz'\right) \\
 -          2           \A{3} \alpha \left(\vecrho_0 \cdot \vecrho_0'\right) \left(\veczeta \cdot \veczeta\right) \\
 +          4           \A{3} \alpha \left(\vecrho_0 \cdot \vecrho_0'\right) \left(\veczeta \cdot \veczeta'\right) \\
 -          2           \A{3} \alpha \left(\vecrho_0 \cdot \vecrho_0'\right) \left(\veczeta' \cdot \veczeta'\right) \\
 -          3           \A{5} \alpha^3 \left(\vecrho_0 \cdot \vecz\right) \left(\vecrho_0 \cdot \vecz'\right) \\
 + \left(   \frac{1}{2} \alpha \A{3} - 3 \alpha^3 \A{5}\right) \left(\vecrho_0 \cdot \vecz\right) \left(\vecrho_0' \cdot \vecz\right) \\
 +          3           \A{5} \alpha^2 \left(\vecrho_0 \cdot \vecz\right) \left(\vecrho_0' \cdot \vecz'\right) \\
 +          3           \A{5} \alpha^2 \left(\vecrho_0 \cdot \vecz'\right) \left(\vecrho_0' \cdot \vecz\right) \\
 + \left( - 3 \alpha \A{5} + \frac{1}{2} \alpha \A{3}\right) \left(\vecrho_0 \cdot \vecz'\right) \left(\vecrho_0' \cdot \vecz'\right) \\
 +          2           \A{3} \alpha \left(\vecrho_0 \cdot \veczeta\right) \left(\vecrho_0' \cdot \veczeta\right) \\
 -          4           \A{3} \alpha \left(\vecrho_0 \cdot \veczeta'\right) \left(\vecrho_0' \cdot \veczeta\right) \\
 +          2           \A{3} \alpha \left(\vecrho_0 \cdot \veczeta'\right) \left(\vecrho_0' \cdot \veczeta'\right) \\
 -          3           \A{5} \alpha \left(\vecrho_0' \cdot \vecz\right) \left(\vecrho_0' \cdot \vecz'\right) \\

\end{align*}

\subsection{Secular disturbing function to degree 4}
\label{sec:A1_secular_ham_4}
We report here the expansion of the secular disturbing function up to degree 4 in the eccentricities and inclinations of the orbits [see Eqs.~\eqref{eq:sec_ham} and \eqref{eq:sec_direct_part}].

\begin{align*}
\left< \frac{a'}{\Delta} \right> = \stepcounter{equation}\tag{\theequation} \\
            \frac{1}{2}         \lapl{1}{0} \\

+             \frac{1}{8}         \lapl{3}{1} \alpha \left( \vecz \cdot \vecz \right) \\
 + \left(  - \frac{1}{2} \lapl{3}{1} + \frac{3}{4} \alpha \lapl{3}{0} - \frac{1}{2} \alpha^2 \lapl{3}{1} \right) \left( \vecz \cdot \vecz' \right) \\
 +          \frac{1}{8}         \lapl{3}{1} \alpha \left( \vecz' \cdot \vecz' \right) \\
 -          \frac{1}{2}         \lapl{3}{1} \alpha \left( \veczeta \cdot \veczeta \right) \\
 +                     \lapl{3}{1} \alpha \left( \veczeta \cdot \veczeta' \right) \\
 -          \frac{1}{2}         \lapl{3}{1} \alpha \left( \veczeta' \cdot \veczeta' \right) \\

+ \left(    \frac{3}{64} \alpha \lapl{5}{1} - \frac{15}{128} \alpha^2 \lapl{5}{0} + \frac{9}{64} \alpha^3 \lapl{5}{1} \right) \left( \vecz \cdot \vecz \right)^2 \\
 + \left(  - \frac{9}{16} \alpha \lapl{5}{1} + \frac{45}{32} \alpha^2 \lapl{5}{0} - \frac{9}{16} \alpha^3 \lapl{5}{1} \right) \left( \vecz \cdot \vecz' \right)^2 \\
 + \left(    \frac{9}{4} \alpha \lapl{5}{1} - \frac{15}{8} \alpha^2 \lapl{5}{0} + \frac{3}{4} \alpha^3 \lapl{5}{1} \right) \left( \vecz \cdot \veczeta \right)^2 \\
 + \left(    \frac{9}{4} \alpha \lapl{5}{1} - \frac{15}{8} \alpha^2 \lapl{5}{0} + \frac{3}{4} \alpha^3 \lapl{5}{1} \right) \left( \vecz \cdot \veczeta' \right)^2 \\
 + \left(    \frac{9}{64} \alpha \lapl{5}{1} - \frac{15}{128} \alpha^2 \lapl{5}{0} + \frac{3}{64} \alpha^3 \lapl{5}{1} \right) \left( \vecz' \cdot \vecz' \right)^2 \\
 + \left(    \frac{3}{4} \alpha \lapl{5}{1} - \frac{15}{8} \alpha^2 \lapl{5}{0} + \frac{9}{4} \alpha^3 \lapl{5}{1} \right) \left( \vecz' \cdot \veczeta \right)^2 \\
 + \left(    \frac{3}{4} \alpha \lapl{5}{1} - \frac{15}{8} \alpha^2 \lapl{5}{0} + \frac{9}{4} \alpha^3 \lapl{5}{1} \right) \left( \vecz' \cdot \veczeta' \right)^2 \\
 + \left(  - \frac{3}{4} \alpha \lapl{5}{1} + \frac{21}{8} \alpha^2 \lapl{5}{0} - \frac{3}{4} \alpha^3 \lapl{5}{1} \right) \left( \veczeta \cdot \veczeta \right)^2 \\
 + \left(  - 3 \alpha \lapl{5}{1} + \frac{21}{2} \alpha^2 \lapl{5}{0} - 3 \alpha^3 \lapl{5}{1} \right) \left( \veczeta \cdot \veczeta' \right)^2 \\
 + \left(  - \frac{3}{4} \alpha \lapl{5}{1} + \frac{21}{8} \alpha^2 \lapl{5}{0} - \frac{3}{4} \alpha^3 \lapl{5}{1} \right) \left( \veczeta' \cdot \veczeta' \right)^2 \\
 + \left(  - \frac{3}{8} \lapl{5}{1} + \frac{15}{16} \alpha \lapl{5}{0} - \frac{27}{32} \alpha^2 \lapl{5}{1} \right) \left( \vecz \cdot \vecz \right) \left( \vecz \cdot \vecz' \right) \\
 + \left(    \frac{9}{32} \alpha \lapl{5}{1} - \frac{27}{64} \alpha^2 \lapl{5}{0} + \frac{9}{32} \alpha^3 \lapl{5}{1} \right) \left( \vecz \cdot \vecz \right) \left( \vecz' \cdot \vecz' \right) \\
 + \left(  - \frac{9}{8} \alpha \lapl{5}{1} - \frac{3}{16} \alpha^2 \lapl{5}{0} - \frac{3}{8} \alpha^3 \lapl{5}{1} \right) \left( \vecz \cdot \vecz \right) \left( \veczeta \cdot \veczeta \right) \\
 + \left(    \frac{9}{4} \alpha \lapl{5}{1} + \frac{3}{8} \alpha^2 \lapl{5}{0} + \frac{3}{4} \alpha^3 \lapl{5}{1} \right) \left( \vecz \cdot \vecz \right) \left( \veczeta \cdot \veczeta' \right) \\
 + \left(  - \frac{9}{8} \alpha \lapl{5}{1} - \frac{3}{16} \alpha^2 \lapl{5}{0} - \frac{3}{8} \alpha^3 \lapl{5}{1} \right) \left( \vecz \cdot \vecz \right) \left( \veczeta' \cdot \veczeta' \right) \\
 + \left(  - \frac{27}{32} \alpha^2 \lapl{5}{1} + \frac{15}{16} \alpha^3 \lapl{5}{0} - \frac{3}{8} \alpha^4 \lapl{5}{1} \right) \left( \vecz \cdot \vecz' \right) \left( \vecz' \cdot \vecz' \right) \\
 + \left(  - \frac{3}{2} \lapl{5}{1} + \frac{15}{4} \alpha \lapl{5}{0} - \frac{9}{8} \alpha^2 \lapl{5}{1} + \frac{15}{4} \alpha^3 \lapl{5}{0} - \frac{3}{2} \alpha^4 \lapl{5}{1} \right) \left( \vecz \cdot \vecz' \right) \left( \veczeta \cdot \veczeta \right) \\
 + \left(    3 \lapl{5}{1} - \frac{15}{2} \alpha \lapl{5}{0} + \frac{9}{4} \alpha^2 \lapl{5}{1} - \frac{15}{2} \alpha^3 \lapl{5}{0} + 3 \alpha^4 \lapl{5}{1} \right) \left( \vecz \cdot \vecz' \right) \left( \veczeta \cdot \veczeta' \right) \\
 + \left(  - \frac{3}{2} \lapl{5}{1} + \frac{15}{4} \alpha \lapl{5}{0} - \frac{9}{8} \alpha^2 \lapl{5}{1} + \frac{15}{4} \alpha^3 \lapl{5}{0} - \frac{3}{2} \alpha^4 \lapl{5}{1} \right) \left( \vecz \cdot \vecz' \right) \left( \veczeta' \cdot \veczeta' \right) \\
 + \left(  - \frac{9}{2} \alpha \lapl{5}{1} + \frac{15}{4} \alpha^2 \lapl{5}{0} - \frac{3}{2} \alpha^3 \lapl{5}{1} \right) \left( \vecz \cdot \veczeta \right) \left( \vecz \cdot \veczeta' \right) \\
 -          \frac{9}{4}         \lapl{5}{1} \alpha^2 \left( \vecz \cdot \veczeta \right) \left( \vecz' \cdot \veczeta \right) \\
 + \left(  - 3 \lapl{5}{1} + \frac{15}{2} \alpha \lapl{5}{0} - \frac{27}{4} \alpha^2 \lapl{5}{1} + \frac{15}{2} \alpha^3 \lapl{5}{0} - 3 \alpha^4 \lapl{5}{1} \right) \left( \vecz \cdot \veczeta \right) \left( \vecz' \cdot \veczeta' \right) \\
 + \left(    3 \lapl{5}{1} - \frac{15}{2} \alpha \lapl{5}{0} + \frac{45}{4} \alpha^2 \lapl{5}{1} - \frac{15}{2} \alpha^3 \lapl{5}{0} + 3 \alpha^4 \lapl{5}{1} \right) \left( \vecz \cdot \veczeta' \right) \left( \vecz' \cdot \veczeta \right) \\
 -          \frac{9}{4}         \lapl{5}{1} \alpha^2 \left( \vecz \cdot \veczeta' \right) \left( \vecz' \cdot \veczeta' \right) \\
 + \left(  - \frac{3}{8} \alpha \lapl{5}{1} - \frac{3}{16} \alpha^2 \lapl{5}{0} - \frac{9}{8} \alpha^3 \lapl{5}{1} \right) \left( \vecz' \cdot \vecz' \right) \left( \veczeta \cdot \veczeta \right) \\
 + \left(    \frac{3}{4} \alpha \lapl{5}{1} + \frac{3}{8} \alpha^2 \lapl{5}{0} + \frac{9}{4} \alpha^3 \lapl{5}{1} \right) \left( \vecz' \cdot \vecz' \right) \left( \veczeta \cdot \veczeta' \right) \\
 + \left(  - \frac{3}{8} \alpha \lapl{5}{1} - \frac{3}{16} \alpha^2 \lapl{5}{0} - \frac{9}{8} \alpha^3 \lapl{5}{1} \right) \left( \vecz' \cdot \vecz' \right) \left( \veczeta' \cdot \veczeta' \right) \\
 + \left(  - \frac{3}{2} \alpha \lapl{5}{1} + \frac{15}{4} \alpha^2 \lapl{5}{0} - \frac{9}{2} \alpha^3 \lapl{5}{1} \right) \left( \vecz' \cdot \veczeta \right) \left( \vecz' \cdot \veczeta' \right) \\
 + \left(    \frac{3}{2} \alpha \lapl{5}{1} - \frac{15}{2} \alpha^2 \lapl{5}{0} + \frac{3}{2} \alpha^3 \lapl{5}{1} \right) \left( \veczeta \cdot \veczeta \right) \left( \veczeta \cdot \veczeta' \right) \\
 + \left(    \frac{3}{2} \alpha \lapl{5}{1} - \frac{3}{4} \alpha^2 \lapl{5}{0} + \frac{3}{2} \alpha^3 \lapl{5}{1} \right) \left( \veczeta \cdot \veczeta \right) \left( \veczeta' \cdot \veczeta' \right) \\
 + \left(    \frac{3}{2} \alpha \lapl{5}{1} - \frac{15}{2} \alpha^2 \lapl{5}{0} + \frac{3}{2} \alpha^3 \lapl{5}{1} \right) \left( \veczeta \cdot \veczeta' \right) \left( \veczeta' \cdot \veczeta' \right) \\

\end{align*}

\subsection{Inequalities of the direct part of the disturbing function}
\label{sec:A1_inequalities}
We consider here the inequalities of the direct part of the disturbing function that are involved in the mean-motion resonances 1:1, 2:1, and 3:2. For each resonance $p\!:\!q$, we report the quantity $2 \real \left( C_{\kappa,\kappa'}^\mathrm{dir} \E^{\iota (\kappa \lambda + \kappa' \lambda')} \right)$, with $C_{\kappa,\kappa'}^\mathrm{dir} = \left< a'/\Delta \, \E^{-\iota (\kappa \lambda + \kappa' \lambda')} \right>$ and $\kappa=-q, \kappa'=p$. The inequalities are expanded up to degree 2 in eccentricities and inclinations of the orbits.  

\begin{align*}
2 \real \left( C_{-1,1}^\mathrm{dir} \E^{\iota (- \lambda + \lambda')} \right) = \stepcounter{equation}\tag{\theequation} \\
                       \lapl{1}{1} \left( \vec\ell \cdot \vec\ell' \right) \\
 + \left(     \lapl{3}{1} - \frac{7}{4} \alpha \lapl{3}{0} +  \alpha^{2} \lapl{3}{1} \right) \left( \vec\ell \cdot \vec\ell' \right) \left( \vecz \cdot \vecz \right) \\
 + \left(  - \frac{5}{3} \alpha^{-1} \lapl{3}{1} + \frac{5}{2} \lapl{3}{0} - \frac{13}{6} \alpha \lapl{3}{1} + \frac{5}{2} \alpha^{2} \lapl{3}{0} - \frac{5}{3} \alpha^{3} \lapl{3}{1} \right) \left( \vec\ell \cdot \vec\ell' \right) \left( \vecz \cdot \vecz' \right) \\
 + \left(     \lapl{3}{1} - \frac{7}{4} \alpha \lapl{3}{0} +  \alpha^{2} \lapl{3}{1} \right) \left( \vec\ell \cdot \vec\ell' \right) \left( \vecz' \cdot \vecz' \right) \\
 + \left(  -  \lapl{3}{1} +  \alpha \lapl{3}{0} -  \alpha^{2} \lapl{3}{1} \right) \left( \vec\ell \cdot \vec\ell' \right) \left( \veczeta \cdot \veczeta \right) \\
 + \left(    2 \lapl{3}{1} - 2 \alpha \lapl{3}{0} + 2 \alpha^{2} \lapl{3}{1} \right) \left( \vec\ell \cdot \vec\ell' \right) \left( \veczeta \cdot \veczeta' \right) \\
 + \left(  -  \lapl{3}{1} +  \alpha \lapl{3}{0} -  \alpha^{2} \lapl{3}{1} \right) \left( \vec\ell \cdot \vec\ell' \right) \left( \veczeta' \cdot \veczeta' \right) \\
 + \left(    \frac{5}{3} \alpha^{-1} \lapl{3}{1} - \frac{5}{2} \lapl{3}{0} + \frac{5}{3} \alpha \lapl{3}{1} - \frac{5}{2} \alpha^{2} \lapl{3}{0} + \frac{5}{3} \alpha^{3} \lapl{3}{1} \right) \left( \vec\ell \cdot \vecz \right) \left( \vec\ell' \cdot \vecz' \right) \\
 + \left(  - \frac{5}{3} \alpha^{-1} \lapl{3}{1} + \frac{5}{2} \lapl{3}{0} - \frac{5}{3} \alpha \lapl{3}{1} + \frac{5}{2} \alpha^{2} \lapl{3}{0} - \frac{5}{3} \alpha^{3} \lapl{3}{1} \right) \left( \vec\ell \cdot \vecz' \right) \left( \vec\ell' \cdot \vecz \right) \\
 + \left(  - 2 \lapl{3}{1} + 4 \alpha \lapl{3}{0} - 2 \alpha^{2} \lapl{3}{1} \right) \left( \vec\ell \cdot \veczeta \right) \left( \vec\ell' \cdot \veczeta' \right) \\
 + \left(    2 \lapl{3}{1} - 4 \alpha \lapl{3}{0} + 2 \alpha^{2} \lapl{3}{1} \right) \left( \vec\ell \cdot \veczeta' \right) \left( \vec\ell' \cdot \veczeta \right) \\

\end{align*}

\begin{align*}
2 \real \left( C_{-1,2}^\mathrm{dir} \E^{\iota (- \lambda + 2\lambda')} \right) = \stepcounter{equation}\tag{\theequation} \\
\left(  - \frac{2}{3} \alpha^{-1} \lapl{3}{1} +  \lapl{3}{0} - \frac{7}{6} \alpha \lapl{3}{1} + \frac{5}{2} \alpha^{2} \lapl{3}{0} - \frac{5}{3} \alpha^{3} \lapl{3}{1} \right) \left( \vec\ell \cdot \vecz \right) \\
 + \left(     \lapl{3}{1} - \frac{5}{2} \alpha \lapl{3}{0} + \frac{3}{2} \alpha^{2} \lapl{3}{1} \right) \left( \vec\ell \cdot \vecz' \right) \\
 + \left(    \frac{4}{3} \alpha^{-1} \lapl{3}{1} - 2 \lapl{3}{0} + \frac{7}{3} \alpha \lapl{3}{1} - 5 \alpha^{2} \lapl{3}{0} + \frac{10}{3} \alpha^{3} \lapl{3}{1} \right) \left( \vec\ell \cdot \vec\ell' \right) \left( \vec\ell' \cdot \vecz \right) \\
 + \left(  - 2 \lapl{3}{1} + 5 \alpha \lapl{3}{0} - 3 \alpha^{2} \lapl{3}{1} \right) \left( \vec\ell \cdot \vec\ell' \right) \left( \vec\ell' \cdot \vecz' \right) \\

\end{align*}

\begin{adjustwidth}{-2cm}{0cm}
\begin{align*}
\stepcounter{equation}\tag{\theequation} \\
2 \real \left( C_{-2,3}^\mathrm{dir} \E^{\iota (-2 \lambda + 3\lambda')} \right) = \\
\left(  - \frac{4}{5} \alpha^{-2} \lapl{3}{1} + \frac{6}{5} \alpha^{-1} \lapl{3}{0} - \frac{31}{30} \lapl{3}{1} + \frac{11}{10} \alpha \lapl{3}{0} - \frac{23}{15} \alpha^{2} \lapl{3}{1} + \frac{16}{5} \alpha^{3} \lapl{3}{0} - \frac{32}{15} \alpha^{4} \lapl{3}{1} \right) \left( \vec\ell' \cdot \vecz \right) \\
 + \left(     \alpha^{-1} \lapl{3}{1} - \frac{3}{2} \lapl{3}{0} + \frac{3}{2} \alpha \lapl{3}{1} - 3 \alpha^{2} \lapl{3}{0} + 2 \alpha^{3} \lapl{3}{1} \right) \left( \vec\ell' \cdot \vecz' \right) \\
 + \left(  - \frac{8}{5} \alpha^{-2} \lapl{3}{1} + \frac{12}{5} \alpha^{-1} \lapl{3}{0} - \frac{31}{15} \lapl{3}{1} + \frac{11}{5} \alpha \lapl{3}{0} - \frac{46}{15} \alpha^{2} \lapl{3}{1} + \frac{32}{5} \alpha^{3} \lapl{3}{0} - \frac{64}{15} \alpha^{4} \lapl{3}{1} \right) \left( \vec\ell \cdot \vec\ell' \right) \left( \vec\ell \cdot \vecz \right) \\
 + \left(    2 \alpha^{-1} \lapl{3}{1} - 3 \lapl{3}{0} + 3 \alpha \lapl{3}{1} - 6 \alpha^{2} \lapl{3}{0} + 4 \alpha^{3} \lapl{3}{1} \right) \left( \vec\ell \cdot \vec\ell' \right) \left( \vec\ell \cdot \vecz' \right) \\
 + \left(    \frac{16}{5} \alpha^{-2} \lapl{3}{1} - \frac{24}{5} \alpha^{-1} \lapl{3}{0} + \frac{62}{15} \lapl{3}{1} - \frac{22}{5} \alpha \lapl{3}{0} + \frac{92}{15} \alpha^{2} \lapl{3}{1} - \frac{64}{5} \alpha^{3} \lapl{3}{0} + \frac{128}{15} \alpha^{4} \lapl{3}{1} \right) \left( \vec\ell \cdot \vec\ell' \right)^{2} \left( \vec\ell' \cdot \vecz \right) \\
 + \left(  - 4 \alpha^{-1} \lapl{3}{1} + 6 \lapl{3}{0} - 6 \alpha \lapl{3}{1} + 12 \alpha^{2} \lapl{3}{0} - 8 \alpha^{3} \lapl{3}{1} \right) \left( \vec\ell \cdot \vec\ell' \right)^{2} \left( \vec\ell' \cdot \vecz' \right) \\

\end{align*}
\end{adjustwidth}




\end{appendices}

\bibliography{sn-bibliography}

\end{document}